\documentclass[11pt]{article}
\usepackage[T1]{fontenc}
\usepackage{amsfonts}
\usepackage{amsmath}
\usepackage{amssymb}
\usepackage{amsthm}
\usepackage{bbm}
\usepackage{bm}
\usepackage{mathrsfs}
\usepackage{verbatim}
\usepackage{setspace}
\usepackage{color}
\usepackage{pdfsync}
\usepackage{enumitem}

\theoremstyle{plain}
\newtheorem{theorem}{Theorem}[section]
\newtheorem{proposition}[theorem]{Proposition}
\newtheorem{lemma}[theorem]{Lemma}
\newtheorem{corollary}[theorem]{Corollary}

\theoremstyle{definition}
\newtheorem{definition}[theorem]{Definition}
\newtheorem{remark}[theorem]{Remark}

\newtheorem{example}[theorem]{Example}

\newtheorem{assumption}[theorem]{Assumption}
\theoremstyle{remark}
\renewenvironment{thebibliography}[1]{%
\begin{oldthebibliography}{#1}%
\setlength{\baselineskip}{.9em}
\linespread{1}%.1
\small
\setlength{\parskip}{0.2ex}%
\setlength{\itemsep}{.4em}%.1
}%
{%
\end{oldthebibliography}%
}
\newcommand{\D}{\mathbb{D}}

\newcommand{\F}{\mathbb{F}}
\newcommand{\G}{\mathbb{G}}

\newcommand{\N}{\mathbb{N}}

\newcommand{\Q}{\mathbb{Q}}
\newcommand{\R}{\mathbb{R}}
\renewcommand{\S}{\mathbb{S}}

\newcommand{\cA}{\mathcal{A}}
\newcommand{\cB}{\mathcal{B}}

\newcommand{\cE}{\mathcal{E}}
\newcommand{\cF}{\mathcal{F}}
\newcommand{\cG}{\mathcal{G}}
\newcommand{\cH}{\mathcal{H}}

\newcommand{\cN}{\mathcal{N}}

\newcommand{\cP}{\mathcal{P}}
\newcommand{\cQ}{\mathcal{Q}}
\newcommand{\cR}{\mathcal{R}}

\newcommand{\cT}{\mathcal{T}}

\newcommand{\cX}{\mathcal{X}}
\newcommand{\cY}{\mathcal{Y}}

\newcommand{\fP}{\mathfrak{P}}

\DeclareMathOperator{\graph}{graph}

\DeclareMathOperator{\tr}{tr}

\DeclareMathOperator{\esssup}{ess\, sup}

\newcommand{\as}{\mbox{-a.s.}}
\newcommand{\qs}{\mbox{-q.s.}}
\newcommand{\1}{\mathbf{1}}
\newcommand{\sint}{\stackrel{\mbox{\tiny$\bullet$}}{}}
\newcommand{\br}[1]{\langle #1 \rangle}

\numberwithin{equation}{section}

\newcommand{\set}[1]{\left\{#1\right\}}
\newcommand{\Real}{\mathbb R}
\newcommand{\Natural}{\mathbb N}
\newcommand{\limn}{\lim_{n \to \infty}}
\newcommand{\nin}{{n \in \Natural}}
\newcommand{\tir}{{t \in \Real_+}}
\newcommand{\Lzp}{\mathbf{L}^0_+}
\newcommand{\Fc}{\mathcal{F}}
\newcommand{\Gc}{\mathcal{G}}
\newcommand{\Stop}{\mathcal{T}}

\newcommand{\Prob}{ {\mathcal{  P}}}
\newcommand{\Qprob}{\mathcal{Q}}
\newcommand{\pare}[1]{\left(#1\right)}
\newcommand{\bra}[1]{\left[#1\right]}
\newcommand{\dbra}[1]{[\kern-0.15em[ #1 ]\kern-0.15em]}
\newcommand{\dbraco}[1]{{[\kern-0.15em[ #1 [\kern-0.15em[ \, }}
\newcommand{\dbraoc}[1]{]\kern-0.15em] #1 ]\kern-0.15em]}
\newcommand{\bF}{\mathbb{F}}
\newcommand{\bG}{\mathbb{G}}
\newcommand{\dfn}{ := }
\newcommand{\simz}{\sim_{\zeta}}
\newcommand{\domz}{\ll_{\zeta}}
\newcommand{\domztom}{\ll_{\zeta^{t,\omega}}}
\newcommand{\Hcal}{\mathcal{H}}
\newcommand{\Hs}{\cH^{\rm simp}}

\def\Pa#1#2{\fP_{#1}(#2)}

\newcommand{\NAP}{{\rm NA}_{1}(\mathcal{P})}
\newcommand{\NA}{{\rm NA}_{1}}

\usepackage[pdfborder={0 0 0}]{hyperref}
\hypersetup{
  urlcolor = black,
  pdfauthor = {Sara Biagini, Bruno Bouchard, Constantinos Kardaras, Marcel Nutz},
  pdfkeywords = {Fundamental Theorem of Asset Pricing; Arbitrage of the First Kind; Superhedging; Nondominated Model},
  pdftitle = {Robust Fundamental Theorem for Continuous Processes},
  pdfsubject = {Robust Fundamental Theorem for Continuous Processes},
  pdfpagemode = UseNone
}
\begin{document}

\title{\vspace{-0em}Robust Fundamental Theorem for Continuous Processes%
}
\date{\today}
\author{
  Sara Biagini%
  \thanks{
  Department of Economics and Management, University of Pisa, Pisa, sara.biagini@ec.unipi.it}
  \and
  Bruno Bouchard%
  \thanks{CEREMADE, Universit\'e Paris Dauphine and CREST-ENSAE, bouchard@ceremade.dauphine.fr. Research supported by ANR Liquirisk and Investissements d'Avenir (ANR-11-IDEX-0003/Labex Ecodec/ANR-11-LABX-0047).
  }
  \and
  Constantinos Kardaras%
  \thanks{Department of Statistics, London School of Economics and Political Science, London, k.kardaras@lse.ac.uk.}
  \and
  Marcel Nutz%
  \thanks{
  Departments of Statistics and Mathematics, Columbia University, New York, mnutz@columbia.edu. Research supported by NSF Grants DMS-1208985 and DMS-1512900. The authors would like to thank the Associate Editor and the anonymous referees for their constructive comments.
  }
 }

\maketitle \vspace{-1em}

\begin{abstract}
We study a continuous-time financial market with continuous price processes under model uncertainty,  modeled via a family $\mathcal{P}$ of possible physical measures. A robust notion ${\rm NA}_{1}(\mathcal{P})$ of no-arbitrage of the first kind is introduced; it postulates that a nonnegative, nonvanishing claim cannot be superhedged for free by using simple trading strategies. Our first main result is a version of the fundamental theorem of asset pricing: ${\rm NA}_{1}(\mathcal{P})$ holds if and only if every $P\in\mathcal{P}$ admits a martingale measure which is equivalent up to a certain lifetime. The second main result provides the existence of optimal superhedging strategies for general contingent claims and a representation of the superhedging price in terms of martingale measures.
\end{abstract}

\vspace{2em}

{\small

}
%\vspace{1em}

\newpage
\section{Introduction}

We consider a financial market where stocks are traded in continuous time. The (discounted) stock price process $S$ is assumed to be continuous, but its distribution in the sense of a stochastic model is not necessarily known. Rather, the market is considered under a family $\cP$ of probability measures: each $P\in\cP$ is understood as a possible model for the real-world dynamics of $S$. Two fundamental questions are studied in this context: the absence of arbitrage and its relation to linear pricing rules (fundamental theorem of asset pricing), and the range of arbitrage-free prices (superhedging theorem).

We introduce a robust notion of market viability, called no-arbitrage of the first kind and denoted $\NAP$. Given a contingent claim $f\geq 0$ at maturity $T$, let   $v^{{\rm simp}}(f)$   be the minimal initial capital necessary to superhedge~$f$ simultaneously under all models $P\in\cP$,
$$
  v^{{\rm simp}}(f) \dfn \inf\big\{x:\,\exists\, H \mbox{ with } x+ H\sint S_{T} \geq f\;P\as\mbox{ for all }P\in\cP  \big\}.
$$
In the above, we allow only simple trading strategies $H$, so that there are no limitations related to defining the stochastic integral $H\sint S$---no semimartingale assumption is made.
Our condition $\NAP$ then postulates that 
$$
  v^{{\rm simp}}(f)=0\quad \mbox{implies} \quad f=0\:\: P\as\mbox{ for all }P\in\cP.
$$
To state the same in reverse, the price $v^{{\rm simp}}(f)$ should be strictly positive if $P \{f>0 \} >0$ holds for some $P\in\cP$.
This condition corresponds to \cite[Definition~1.1]{Kardaras.10} when $\cP$ is a singleton; it will turn out to be a notion of market viability that is well suited for model uncertainty in continuous time.

The main goal of the fundamental theorem is to deduce the existence of martingale measures, or linear pricing rules, from the absence of arbitrage opportunities. In the classical case \cite{DalangMortonWillinger.90,DelbaenSchachermayer.94}, this measure is equivalent to the physical measure $P$. In the case of model uncertainty in a discrete-time market, the fundamental theorem of \cite{BouchardNutz.13} yields a family $\cQ$ of martingale measures such that each $P\in\cP$ is dominated by a martingale measure; the families $\cP$ and $\cQ$ are equivalent in the sense that they have the same polar sets. In the present setting with continuous processes, we find a result which is stronger in the sense that each $P$ admits an equivalent martingale measure~$Q$. On the other hand, equivalence needs to be defined in a weaker way: it is necessary to allow for a loss of mass in our martingale deflators; thus, the measures $Q$ may allocate mass outside the support of $\cP$. As a result, the equivalence of measures holds only up to a random time $\zeta$, and so does the martingale property. More precisely, we suppose that our model is set on a canonical space $\Omega$ of paths which are continuous before possibly jumping to a cemetery state, and $\zeta$ is the time of this jump. This ``lifetime'' is infinite and thus invisible under all $P\in\cP$, but may be finite under some $Q\in\cQ$. With these notions in place, our version of the fundamental theorem then states that $\NAP$ holds if and only if for every $P\in\cP$ there exists a local martingale measure $Q$ such that $Q$ and $P$ are equivalent prior to $\zeta$. See Definition~\ref{defn:ELMM} and Theorem~\ref{thm: FTAP} for the precise statements.

A related setting is considered in \cite{DolinskySoner.12} where $S$ is the canonical process in the space of continuous paths. Roughly speaking, the market model considered there corresponds to declaring all paths to be possible for the stock price, or including all measures in $\cP$. There is, then, no necessity for a definition of arbitrage; in some sense, the absence of the latter is implicit in the fact that all paths are possible. Nevertheless, the duality result stated in \cite{DolinskySoner.12} implies a conclusion in the direction of the fundamental theorem; namely, it follows that there must exist at least one martingale measure under the conditions of that result. A similar result on Skorokhod space is reported in~\cite{DolinskySoner.14}. We also refer to \cite{DavisHobson.07} for a discussion of different notions of arbitrage in the context of traded options. For versions of the robust fundamental theorem for discrete-time frictionless markets, see \cite{AcciaioBeiglbockPenknerSchachermayer.12, BouchardNutz.13, BurzoniFrittelliMaggis.14, Riedel.11}; for discrete-time markets with transaction costs, see \cite{BayraktarZhang.13, BayraktarZhangZhou.13, BouchardNutz.14trans, DolinskySoner.13}.

The second main result of the present paper is a superhedging theorem in our setting. Assume that $\NAP$ holds and let $f\geq0$ be a contingent claim, measurable at time $T$. Then, we establish the duality
\begin{multline*}
  \sup_{Q\in\cQ}  E^Q[f\1_{\zeta>T}] \\
  =\inf\big\{x:\,\exists\, H \mbox{ with } x+ H\sint S_{T} \geq f\;P\as\mbox{ for all }P\in\cP  \big\};
\end{multline*}
moreover, we construct an optimal superhedging strategy $H$---naturally, this necessitates continuous trading. See Theorem~\ref{th:duality} for the precise statement. 

The line of argument in the proof is similar to \cite{Nutz.14} where it is assumed that $\cP$ consists of martingale measures in the first place. In the present case, the martingale property holds only prior to $\zeta$ which necessitates a number of additional considerations. Generally speaking, the superhedging theorem is fairly well studied in the situation where $\cP$ consists of martingale measures; cf.\ \cite{BouchardMoreauNutz.12, DenisMartini.06, FernholzKaratzas.11, NeufeldNutz.12, NutzSoner.10, NutzZhang.12, Peng.10, PossamaiRoyerTouzi.13, SonerTouziZhang.2010rep, SonerTouziZhang.2010dual}, among others, or when  all paths are considered possible for the stock and options are also traded; see, e.g., \cite{CoxObloj.11,DavisHobson.07,DolinskySoner.12,DolinskySoner.14,GalichonHenryLabordereTouzi.11,Hobson.98}. We also refer to \cite{AcciaioBeiglbockPenknerSchachermayer.12, BayraktarZhangZhou.13, BouchardNutz.13, DolinskySoner.13, Nutz.13} for discrete-time markets. Finally, in   the   forthcoming independent work~\cite{CheriditoKupperTangpi.14}, absence of a duality gap will be established by functional analytic methods in a market more general than ours, under a condition that is stronger than $\NAP$.
 
The remainder of this paper is organized as follows. The setup is detailed in Section~\ref{sec: set up}, where we also define $\NAP$. In Section~\ref{sec : FTAP}, we discuss our version of the fundamental theorem of asset pricing. Section~\ref{sec: dyna prog and local mart} provides some technical results on prior-to-$\zeta$ equivalent martingale measures; these are used in Section~\ref{sec: superhedging duality}, where we study the superhedging theorem. Finally, the Appendix collects auxiliary results on F\"ollmer's exit measure and the particular path space that are used in the body of this paper.

\section{Setup} \label{sec: set up}

\subsection{Measurable Space and Model Uncertainty}\label{subsec: prelims}

We first construct the underlying measurable space $(\Omega,\Fc)$ used throughout the paper. Let $E$ be a Polish space and let $d_E$ be a complete metric consistent with the topology of $E$. Adjoining an isolated ``cemetery'' state~$\triangle$, we shall work with $\bar E:=E\cup \{\triangle\}$. It is easy to see that  $\bar E$ is again a Polish space under the metric
$$
d_{\bar E}(x,y):= 1\wedge d_{E}(x,y) \1_{\{\triangle \notin\{x,y\}\}}+\1_{\{\triangle \in\{x,y\}\}\cap \{x\ne y\}},\quad x,y\in\bar E.
$$
We then define $\Omega$ to be the space of all paths $\omega : \R_{+} \to \bar E$ which start at a 
given point $x_{*}\in E$, are c\`adl\`ag on $[0,\zeta(\omega))$ and constant on $[\zeta(\omega),\infty)$, where 
$$
  \zeta(\omega):=\inf\{t\geq0: \,\omega_{t}=\triangle\}
$$
is the ``lifetime'' of $\omega$. The function $\zeta$ takes values in $(0,\infty]$ since $x_{*}\in E$ and the paths are right-continuous.  It is shown in  Lemma~\ref{le:OmegaPolish} in the Appendix  that $\Omega$ carries a natural Polish topology.

We denote by $B=(B_{t})_{t\in\R_{+}}$ the canonical process, defined by $B_{t}(\omega)=\omega_{t}$, and by $\F=(\cF_{t})_{t\in\R_{+}}$ its the natural filtration, $\cF_{t}=\sigma(B_{s},\,s\leq t)$, and finally $\cF=\sigma(B_{s},\,s\in\R_{+})$. The set of $\F$-stopping times is denoted by $\Stop$. The minimal right-continuous filtration containing $\F$ is denoted by $\bF_+ = (\Fc_{t+})_{\tir}$, while $\Stop_+$ is the set of all $\bF_+$-stopping times. With these notions in place, we observe that $ \{ \zeta \leq t \} = \{B(t) = \triangle \} \in \Fc_t$ for all $\tir$ and hence that $\zeta \in \Stop$.
 
%\begin{remark}  \label{rem:sigma-field}
%The $(0, \infty)$-valued random variable $\zeta$ is such that $\zeta \in \Stop$---indeed, $ \{ \zeta \leq t \} = \{ \omega(t) = \triangle \} \in \Fc_t$ for all $\tir$ follows from the definition of $\zeta$. A straightforward argument shows that $\Fc_{\zeta -} = \Fc_{\zeta} =\cF$.
%\end{remark}

To represent model uncertainty, we shall work with a (nonempty) family~$\cP$ of probability measures on $(\Omega,\Fc)$, rather than a single measure. Each element $P\in\cP$ is interpreted as a possible model for the real-world dynamics; no domination assumption is made. We say that a property holds $\cP$-quasi-surely (or $\cP$-q.s.)  if it holds $P$-a.s.\ for all $P\in \cP$. We shall assume throughout that
\[
\zeta = \infty \quad \Prob\qs
\]
Thus, the cemetery state is actually invisible under the real-world models; its role will be to absorb the residual mass of certain martingale measures.

Given a $\sigma$-field $\Gc \subseteq \Fc$, we denote by $\Lzp(\Gc)$ the set of all $[0, \infty]$-valued, $\Gc$-measurable random variables that are $\Prob$-q.s.\ finite.

\subsection{Trading and Arbitrage}

The tradable assets are modeled by an $\R^{d}$-valued, $\bF$-adapted and right-continuous process  $S : \R_{+}\times\Omega \to \Real^d$ such that 
$$ 
\mbox{the paths of $S$ are $\Prob$-q.s.\ continuous.}
$$
No other assumption is made on $S$ at this stage; in particular, no semimartingale property is assumed. However, structural  properties will follow later as a consequence of our no-arbitrage condition.

A simple predictable\footnote{We define simple predictable strategies with respect to the filtration $\bF_+$; however, we recall that the class of predictable processes on $(\Omega, \bF)$ coincides with the   class of predictable processes on $(\Omega, \bF_+)$. The symbol $\dbraoc{\tau_{i-1}, \tau_i}$ denotes the stochastic interval.}
strategy is a process $H = \sum_{i = 1}^n h_i \1_{\dbraoc{\tau_{i-1}, \tau_i}}$, where $h_i = (h_i^j)_{j\le d}$ is $\Fc_{\tau_{i-1}+}$-measurable for all $i\le n$, and $(\tau_i)_{i\le  n}$ is a nondecreasing $\Stop_+$-valued sequence with $\tau_0 = 0$. Given an initial capital $x \in \Real_{+}$ and a simple predictable strategy $H$, we define the associated wealth process
\[
X^{x, H} = x + H\sint S = x + \sum_{i=1}^n \sum_{j=1}^d h_i^j \pare{S^j_{\tau_i \wedge \cdot} - S^j_{\tau_{i-1} \wedge \cdot}}.
\]
Moreover, we define  $\cH^{\rm simp}(x)$ as the class of all simple predictable processes $H$ such that $X^{x, H}$ remains nonnegative $\Prob$-q.s. (The superscript  ``${\rm simp}$''  acts as a mnemonic for ``simple'' in what follows.)
Given   $T \in \Real_+$ and $f \in \Lzp(\Fc_T)$, let 
\[
v^{{\rm simp}}(T, f) \dfn \inf \set{ x \in \Real_+ : \, \exists H \in \Hs(x) \text{ with } X^{x, H}_T \geq f \: \Prob \text{-q.s.}}
\]
be the superhedging price of the claim $f$ over the class of simple strategies. We can then introduce our notion of no-arbitrage of the first kind, stating that the superhedging price is null if and only if the claim is null $\cP$-q.s.
%Then, our condition $\NAP$ reads as follows.

\begin{definition}\label{def: NA1(Pc)} We say that  $\NAP$  holds if
\[
\forall T \in \Real_+ \text{ and } f \in  \Lzp(\Fc_T), \quad v^{{\rm simp}}(T, f) = 0 \;\Longrightarrow\; f = 0 \; \Prob \text{-q.s.}
\] 
\end{definition}

This condition coincides with \cite[Definition~1.1]{Kardaras.10} when $\cP$ is a singleton.

\section{Fundamental Theorem of Asset Pricing}\label{sec : FTAP}
 
%\subsection{Prior-to-$\zeta$ equivalent measures}

In order to state our version of the fundamental theorem of asset pricing, we first need to introduce the concept of prior-to-$\zeta$ equivalence.
%Probability measures that are equivalent for all times strictly prior to $\zeta$ will play an important role in our subsequent treatment.

\begin{definition}\label{def : prior to zeta abs cont and equiv}
Given two measures $P$ and $Q$ on $(\Omega, \Fc)$, we say that $Q$ is   \emph{prior-to-$\zeta$ absolutely continuous} with respect to $P$, if $Q\ll P$ holds on the space $\pare{\set{t < \zeta}, \Fc_t \cap \set{t < \zeta}}$ for all $\tir$.   This relation is denoted by $Q \domz P$.  If $Q \domz P$ and $P \domz Q$, we say that $P$ and $Q$ are \emph{prior-to-$\zeta$ equivalent} and denote this fact by $Q \simz P$.  
\end{definition}

In this definition, equivalence is used in the sense of unnormalized measures. Namely, even if the measures are probabilities on $(\Omega,\cF)$, they need not be probabilities on $\pare{\set{t < \zeta}, \Fc_t \cap \set{t < \zeta}}$, and $Q \simz P$ does not mean that $P (A) = 1$ implies  $Q (A) = 1$, even if $A\in \Fc_t \cap \set{t < \zeta}$. A second remark is that local (on $\Fc_t$, for all $\tir$) equivalence of two probabilities trivially implies prior-to-$\zeta$ equivalence, but the converse fails. The following simple example demonstrates these phenomena.

\begin{example}
Suppose that $E$ is a singleton. Then, $\bF$ is the smallest filtration that makes $\zeta$ a stopping time and $\Fc_t \cap \set{t < \zeta} = \set{\emptyset, \set{t < \zeta}}$ holds for all $\tir$. It follows that prior-to-$\zeta$ equivalence for any two probabilities $P$ and $Q$ on $(\Omega, \Fc)$ is tantamount to checking that $P\{\zeta > t\} > 0$  if and only if $ Q\{\zeta > t\} > 0$, for all $\tir$. On $(\Omega, \Fc)$, one can prescribe probabilities endowing any given law to $\zeta$; letting $P$ be such that $P\{\zeta < \infty\} = 0$ and $Q$ be such that $Q\{\zeta > t\} = \exp(-t)$ for $\tir$, it follows that $P$ is a probability on $\pare{\set{t < \zeta}, \Fc_t \cap \set{t < \zeta}}$ for all $t \in (0, \infty)$, while $Q$ is a strict sub-probability. Note also that the probabilities $P$ and $Q$ fail to be equivalent on $\Fc_t$ whenever $t \in (0, \infty)$; indeed, $P\{\zeta \leq t\} = 0$ and $Q\{\zeta \leq t\} > 0$ hold for all $t \in (0, \infty)$.
\end{example}

We refer to Section~\ref{sec: follmer measure} for further discussions on prior-to-$\zeta$ equivalence and proceed with the relevant concept of a local martingale measure. 

\begin{definition} \label{defn:ELMM}
Fix $P \in \Prob$. A probability $Q$ on $(\Omega, \Fc)$ is a \emph{prior-to-$\zeta$ equivalent local martingale measure  corresponding to $P$} if $Q \simz P$ and there exists a nondecreasing sequence $({\tau}_n)_{\nin}\subset \Stop_+$ such that
\begin{enumerate} 
	\item $\tau_n < \zeta$ for all $\nin$ and $\limn \tau_n = \zeta$ hold $Q$-a.s.,
	\item $(S_{t\wedge \tau_n })_{\tir}$ is an $(\bF_+, Q)$-martingale for all $\nin$.
\end{enumerate}
The class of all such probabilities $Q$ will be denoted by $\Qprob^P$.
\end{definition}

What follows is the main result of this section, the fundamental theorem of asset pricing. In the present incarnation, it states that the condition $\NAP$ of Definition~\ref{def: NA1(Pc)} holds if and only if we can find (at least) one prior-to-$\zeta$ equivalent local martingale measure for each possible model $P\in \cP$. 

\begin{theorem}\label{thm: FTAP}
Condition $
\NAP$ holds if and only if $\Qprob^P \neq \emptyset$ for all $P \in \Prob$.
\end{theorem}

We  emphasize that this result necessitates the continuity of $S$; it is to be compared to the discrete-time case of \cite{BouchardNutz.13}. The following is a direct consequence of the theorem, but will actually be established in the course of its proof.

\begin{corollary} \label{co:FTAP}
  Let $\NAP$ hold. Then $S$ is a semimartingale under each $P\in \cP$.
\end{corollary}

To be precise, we should indicate a filtration in the above statement. In fact, the $P$-semimartingale property holds equivalently in any of the filtrations $\F$, $\F_{+}$ or $\F_{+}^{P}$  (the $P$-augmentation of $\F_{+}$), or more generally in any intermediate filtration $\F\subset \G \subset \F_{+}^{P}$; see, e.g., \cite[Proposition~2.2]{NeufeldNutz.13a}. We shall use this fact in Section~\ref{sec: superhedging duality}.

\begin{proof}[Proof of Theorem~\ref{thm: FTAP}]
\emph{Step 1.} We first prove the easy implication; that is, we assume that $\Qprob^P \neq \emptyset$ for all $P\in\cP$.  
Fix $T \in \Real_+$ and $f \in \Lzp(\Fc_T)$ with $v^{{\rm simp}}(T, f) = 0$. Moreover, let $P\in\cP$ be arbitrary but fixed; we need to show that $f=0$ $P$-a.s.

Indeed, let $\cX^{{\rm simp}}$ be the class of all processes of the form $X^{x, H}$ for $x \in \Real_+$ and $H \in \cH^{\rm simp}(x)$. By assumption, there exists some  $Q\in \Qprob^{P}$. Let $(\tau_n)_{\nin}$ be the localizing sequence appearing in Definition~\ref{defn:ELMM}. Since the stopped process $S_{\cdot \wedge \tau_n}$ is a $Q$-martingale, it follows that $X_{\cdot \wedge \tau_n}$ is a local $Q$-martingale for all $X \in \cX^{{\rm simp}}$ and $\nin$. A straightforward argument then shows that $X \1_{[\![0, \zeta[\![}$ is a $Q$-supermartingale for all $X \in\cX^{{\rm simp}}$.

Let $X^n\in \cX^{{\rm simp}}$ be such that $X_0^n = 1/n$ and $X^n_T \geq f$ $\Prob$-q.s., then the above supermartingale property yields that 
\[
  E^Q [f \1_{T < \zeta}] \leq E^Q[X^n_T \1_{T < \zeta}] \leq E^Q[X^n_0 ] = 1/n,\quad n\geq1.
\]
Therefore, $E^Q [f \1_{T < \zeta}] = 0$ which implies that $Q \{f > 0, T < \zeta\} = 0$. Since $Q \simz P$ and $\zeta=\infty$ $P$-a.s., it follows that $P \{f > 0\} = 0$. This completes the proof of the ``if'' implication in Theorem~\ref{thm: FTAP}.

\emph{Step 2.}  The converse implication will be established through a third equivalent condition. To this end, consider $\NA(P):=\NA(\{P\})$ for a fixed $P\in\cP$; that is, the condition that 
  \[
\forall T \in \Real_+ \text{ and } f \in  \Lzp(\Fc_T), \quad v^{{\rm simp},P}(T, f) = 0 \;\Longrightarrow\; f = 0 \; P \text{-a.s.},
\] 
where
\[
v^{{\rm simp},P}(T, f) = \inf \big\{ x \in \Real_+ : \, \exists H \in \cH^{{\rm simp},P}(x) \text{ with } X^{x, H}_T \geq f \, P \text{-a.s.}\big\}
\]
and $\Hcal^{{\rm simp},P}(x)$ is the class of all simple predictable processes $H$ such that $X^{x, H}$ is nonnegative $P$-a.s.
We claim that 
\begin{equation} \label{eq: na_1_easy}
\text{$\NAP$ holds~~~~if and only if~~~~$\NA(P)$ holds for all $P \in \Prob$.}
\end{equation}
Indeed, the observation that $\cH^{\rm simp}(x) \subseteq \Hcal^{{\rm simp},P}(x)$ shows that  the validity of $\NA(P)$ for all $P \in \Prob$ implies $\NAP$. To see the converse, suppose that there exists $P \in \Prob$ such that $\NA(P)$ fails. Then, there are $T \in \Real_+$ and $g \in  \Lzp(\Fc_T)$ such that $v^{{\rm simp},P}(T, g) = 0$ and $P\{ g > 0\}>0$. That is, for any $\nin$ there exists $H^n \in \Hcal^{{\rm simp},P} (1/n)$ such that $X^{1/n, H^n}_T \geq g$ $P$-a.s. Define
\[
  \tau^n = \inf \big \{ \tir :\, X^{1/n, H^n}_t < 0 \big \}\in \Stop_+,\quad G^{n} = H^n \1_{\dbraoc{0, \tau^n}}.
\]
Then $\tau^n \in \Stop_+$ as the paths of $S$ are right-continuous and thus $G^{n}$ is a simple predictable strategy. Since $\tau^{n}=\infty$ $P$-a.s., we have $G^{n} = H^n$ $P$-a.s.; in particular, $G^{n}$ still satisfies $X^{1/n, G^n}_T \geq g$ $P$-a.s. In addition, the definition of $\tau^{n}$ guarantees that $X^{1/n, G^n}$ is nonnegative $\cP$-q.s.---the continuity of $S$ is crucial in this step. Consider 
\[
  f := \inf_{\nin} X^{1/n, G^n}_T \in \Lzp(\Fc_T)
\]
and note that $v^{{\rm simp}}(T, f) = 0$ holds by definition. Moreover, we have $f \geq g$ $P$-a.s.\ and thus $P\{f  > 0\} > 0$, contradicting $\NAP$. Therefore, \eqref{eq: na_1_easy} has been established.

\emph{Step 3.} In view of \eqref{eq: na_1_easy}, it remains to show that $\NA(P)$ implies $\cQ^{P}\neq\emptyset$, for arbitrary but fixed $P\in\cP$. Thus, we are essentially in the realm of classical stochastic analysis and finance; in particular, we may use the tools in the Appendix as well as \cite{Kardaras.10,Kardaras.13}. 

Define $\cX^{{\rm simp},P}$ as the class of all processes of the form $X^{x, H}$ for $x \in \Real_+$ and $H \in \cH^{{\rm simp},P}(x)$. The set $\{X \in \cX^{{\rm simp},P}:\, X_0 = 1 \}$ has the essential properties of \cite[Definition 1.1]{Kardaras.13} needed to conclude that $\cX^{{\rm simp},P}$ consists of $P$-semimartingales, see \cite[Theorem 1.3]{Kardaras.13}, and that (the immediate extension of) condition  $\NA(P)$ is also valid for the closure $\cX^{P}$ of $\cX^{{\rm simp},P}$ in the $P$-semimartingale topology; see \cite[Remark 1.10]{Kardaras.13}. In particular, a standard localization and integration argument (using local boundedness of $S$ under~$P$) shows that $S$ is itself a $P$-semimartingale. 

The set $\cX^{P}$ coincides with the class of all $P$-a.s.\ nonnegative stochastic integrals of $S$ under $P$, using general predictable and $S$-integrable integrands. This is seen by using density (in the semimartingale topology) of simple stochastic integrals with respect to general stochastic integrals, as well as a stopping argument which again uses that $S$ has continuous paths $P$-a.s.
As a result, using condition $\NA(P)$ for $\cX^{P}$, we infer the existence of a strictly positive $(\bF_+, P)$-local martingale $Y$ with $Y_0 = 1$ such that $Y S$ is an $(\bF_+, P)$-local martingale; cf.\ \cite[Theorem~4]{Kardaras.10}. We can now use Theorem~\ref{thm: foelmeasure_1} in the Appendix to construct a probability $Q \simz P$ such that $Y$ is the prior-to-$\zeta$ density of $Q$ with respect to $P$. Using the facts that $Y S$ is an $(\bF_+, P)$-local martingale, $\zeta$ is foretellable under $Q$ (for the latter, see Definition~\ref{defn: foretellable} and Theorem~\ref{thm: foelmeasure_1} in the Appendix) and Remark~\ref{rem:localization_under_both}, we can construct the required $\Stop_+$-valued sequence $(\tau_n)_{\nin}$ such that $S_{\cdot \wedge \tau_n}$ is an $(\bF_+, Q)$-martingale for all $\nin$. The last fact translates to $Q \in \Qprob^P$ and concludes the proof.
\end{proof}

%%%%%%%%%%%%%%%%%%%%%%%%%%%%%%%%%%%%%%%%%%%%%%%%%%%%%%%%%%%%%%%%%%%%%%%%%%%%%%%%%%%%%%%%%%%%%%%%%%%%

 %%%%%%%%%%%%%%%%%%%%%%%%%%%%%%%%%%%%%%%%%%%%%%%%%%%
\section{Dynamic Programming Properties of Prior-to-$\zeta$ Supermartingale Measures}\label{sec: dyna prog and local mart}

For our proof of the superhedging theorem in Section~\ref{sec: superhedging duality}, it will be crucial to know that the set of (super-)martingale measures satisfies certain dynamic programming properties. In this section, we impose assumptions on the set~$\cP$ which is the primary object of our model, and show how these properties are inherited by the corresponding set of supermartingale measures.
%
%
%structural  
%In order to provide a superhedging theorem in our setting, we follow the reasoning in \cite{Nutz.14}; in the latter paper, the family of models, denoted by $\cQ$, already consists of martingale measures. The superhedging theorem then relies on the super-martingale property, under each $Q\in \cQ$, of the candidate superhedging price process. Writing its Doob-Meyer decomposition under any $Q \in \Qprob$ allows one to construct a superhedging strategy, which, by identification of its quadratic variation under the different $Q$, does not depend on the model-at-hand. See Section~\ref{sec: superhedging duality} below for details. The aforementioned super-martingale property is established by use of dynamic programming arguments; as is typical, such arguments require certain natural assumptions on the set of martingale measures, such as the ability to ``paste'' models together. In this work, since the starting point is the collection $\cP$ of scenarios, it is natural to impose such structural conditions directly on $\cP$---see Definition~\ref{def : analytic and stable family} and Assumption \ref{ass : Pc(t,omega)} below. It is then crucial to show that our collection of prior-to-$\zeta$ local martingale measures inherits these properties; this task is taken up in Section~\ref{sec: analytic and stable}.

\subsection{Additional Assumptions and Notation}\label{subsec: add condi and notations}

From now on, we assume that the Polish space $E$ is a topological vector space and that the paths $\omega\in\Omega$ start at the point $x_{*}=0\in E$.

For $x,y\in \bar E$, we use the convention $x+y=\triangle$ if $x=\triangle$ or $y=\triangle$.
Let $t\geq0$. Given $\omega,\tilde\omega\in\Omega$, we set 
 $$
 (\omega\otimes_{t}\tilde \omega)_{s}=\omega_{s}\1_{[0,t)}(s)+ (\omega_{t}+\tilde \omega_{s-t})\1_{[t,\infty)}(s).
 $$
 Given also a process $Z$,  we define
 $$
 Z_{s}^{t,\omega}(\tilde \omega):=Z_{t+s}(\omega\otimes_{t}\tilde \omega), \quad s\ge 0; 
 $$
 note that a shift in the time variable is part of our definition. We view a random variable $\xi$ as a process which is constant in time, so that 
 $$
 \xi^{t,\omega}(\tilde \omega):=\xi(\omega\otimes_{t}\tilde \omega). 
 $$
  We denote by $\fP(\Omega)$ the collection of all probability measures on $\Omega$, equipped with the topology of weak convergence.  Given a probability $R \in \fP(\Omega)$, we  define $R^{t,\omega}$ by 
 $$
 R^{t,\omega}(A)=R_t^\omega\{\omega \otimes_{t} \tilde \omega: \tilde \omega \in A\}, \quad A\in \cF, 
 $$
 where $R_t^\omega$ is a regular conditional distribution of $R$ given $\cF_{t}$  satisfying
 $$
 R_t^\omega \{\omega'\in \Omega: \omega'=\omega\mbox{ on } [0,t]\} =1, \quad \omega\in \Omega. 
 $$
 The existence of $R_t^\omega$ is guaranteed by the fact that $\cF_{t}$ is countably generated; cf.\ Lemma~\ref{le:OmegaPolish} and \cite[Theorem~1.1.8 and p.\,34]{StroockVaradhan.79}. It then follows that
 \begin{equation}\label{eq:conditioning}
  E^{R^{t,\omega}}[\xi^{t,\omega}]=E^{R^{\omega}_{t}}[\xi]=E^{R}[\xi|\cF_{t}](\omega)\quad\mbox{for $R$-a.e.\ $\omega\in\Omega$}.
 \end{equation}

We shall assume that our set $\cP$ admits a family of $(t,\omega)$-conditional models. More precisely, we start with a family
$\{\cP_{t}(\omega):\,  t\in\R_{+},\,\omega \in \Omega\}$ of subsets of $\fP(\Omega)$  which is  adapted in the sense that $\cP_{t}(\omega)=\cP_{t}(\tilde\omega)$ if $\omega|_{[0,t]}=\tilde\omega|_{[0,t]}$. In particular, $\cP_{0}=\cP_{0}(\omega)$ is independent of $\omega$. We impose the  following structural conditions---compare with \cite{NeufeldNutz.12,NutzVanHandel.12} in the case $\zeta\equiv \infty$.

\begin{definition}\label{def : analytic and stable family} An adapted  family $\{\cR_{t}(\omega):\,  t\in\R_{+},\,\omega \in \Omega\}$ of subsets of $\fP(\Omega)$  is \emph{analytic and stable prior to $\zeta$} if the following hold
for all  $ t\ge s\ge0$, $\bar \omega \in \Omega$ and $  R\in \cR_{s}(\bar \omega)$.
  \begin{enumerate}

 \item[(A1)]   $\{(R',\omega): \omega \in \Omega, R'\in \cR_{t}(\omega)\}\subset \fP(\Omega)\times \Omega$ is analytic\footnote{The definition of an analytic set is recalled in Section~\ref{sec:measureTheory} of the Appendix.}.

 \item[(A2)] $R^{t-s,\omega} \in \cR_{t}(\bar \omega \otimes_{s} \omega)$ for $R$-a.e.~$\omega \in \{\zeta^{s,\bar \omega}>t\}$. 
 
 \item[(A3)] If $\nu:\Omega\mapsto \fP(\Omega)$ is an $\cF_{t-s}$-measurable kernel and $\nu(\omega)\in \cR_{t}(\bar \omega\otimes_{s} \omega)$ for $R$-a.e.~$\omega \in \{\zeta^{s,\bar \omega}>t\}$, then the measure defined by 
 $$
\bar R(A):=\iint (\1_{A})^{t-s,\omega}(\omega')\,\nu^{R}(d\omega';\omega)\,R(d\omega),\quad A\in \cF,
 $$
 $$
 \mbox{where}\quad  \nu^{R}(\omega):=\nu(\omega)\1_{\{\zeta^{s,\bar \omega}>t\}}(\omega) + R^{t-s,\omega}\1_{\{\zeta^{s,\bar \omega}\leq t\}}(\omega),
 $$
 belongs to $\cR_{s}(\bar \omega)$.
 \end{enumerate}
\end{definition}

Condition (A1) is of technical nature; it will be used for measurable selection arguments. Conditions~(A2) and~(A3) are natural consistency conditions, stating that the family is stable under ``conditioning'' and ``pasting.'' 
  
\begin{assumption}\label{ass : Pc(t,omega)} 
  We have $\cP=\cP_{0}$ for a family $\{\cP_{t}(\omega):\, t\in\R_{+},\,\omega \in \Omega\}$ which is analytic and stable prior to $\zeta$. Moreover, $S^{t,\omega}$ is $\cP_{t}(\omega)$-q.s.\ con\-tinuous  prior to $\zeta^{t,\omega}-t$, for all $t\in\R_{+}$ and $\omega \in \Omega$.
\end{assumption}

A canonical example of such a set $\cP$ is the collection of all laws $P$ of It\^o semimartingales $\int_0^\cdot \alpha_u \,du + \int_0^\cdot \sigma_u \,dW_u$, each one situated on its own probability space with a Brownian motion $W$, drift rate $\alpha$ valued in a given measurable set $A\subset \R^{d}$, and volatility $\sigma$ such that $\sigma\sigma^{\top}$ is valued in a given measurable set $\Sigma$ of positive definite $d\times d$ matrices. In this case, we can take $\cP_{t}(\omega)=\cP$ for all $(t,\omega)$ because the sets $A$ and $\Sigma$ are constant; cf.~\cite{NeufeldNutz.13b}. The continuity condition is clearly satisfied for the canonical choice $S=B$ and then $\NA(\cP)$ holds, for instance, when $A$ and $\Sigma$ are compact.
%
%\smallskip
%
%The next Subsection shows that these properties are, in essence, inherited by the corresponding family of prior-to-$\zeta$ absolutely continuous supermartingale measures $\Qprob$. 
%%While the fact that $\zeta\equiv \infty$ holds $\cP$-q.s. makes the sets related to $\zeta$ in Definition~\ref{def : analytic and stable family} appear superfluous, one does not necessarily have that $\zeta = \infty$ $\Qprob$-q.s., which makes the presence of $\zeta$ crucial.
%While $\zeta\equiv \infty$ holds $\cP$-q.s., it is not necessary that $\zeta = \infty$ holds $\Qprob$-q.s.; therefore, the sets related to $\zeta$ in Definition~\ref{def : analytic and stable family}, which may appear superfluous under $\Prob$, are essential for passing from $\Prob$ to $\Qprob$.

\subsection{Prior-to-$\zeta$ Supermartingale Measures  }\label{sec: analytic and stable}
 
For technical reasons, it will be convenient to work with supermartingale (rather than local martingale)  measures in what follows. The purpose of this section is to define a specific family of supermartingale measures satisfying the conditions of Definition~\ref{def : analytic and stable family}; it will be used to construct the optimal strategy in the superhedging theorem (Theorem~\ref{th:duality}). We first need to define a conditional notion of prior-to-$\zeta$ absolute continuity. 

\begin{definition}
  Let $(t,\omega)\in \R_{+}\times\Omega$ and $P,Q\in \fP(\Omega)$. We write $Q\domztom P$ (with some abuse of notation) if
 \[
   Q\ll P \quad\mbox{on}\quad \cF_{s}\cap \{s<\zeta^{t,\omega}-t\},\quad s\in\R_{+}.
 \]
\end{definition}

We also need to consider wealth processes conditioned by $(t,\omega)\in\R_{+}\times \Omega$. More precisely, let  
\begin{equation}\label{eq:defXsimple}
  \cX_{t}^{\rm simp}(\omega) := \big\{1 +   (H\sint S^{t,\omega})^{\tau^{n}_{H,S^{t,\omega}}}:\,H\in\cH^{\rm simp} , \, n\in\N\big\},
\end{equation}
where $\cH^{\rm simp}$ is the set of all simple predictable processes
and 
$$
  \tau^{n}_{H,S^{t,\omega}}:=\inf\big\{ s\ge 0  : (H\sint S^{t,\omega})_{s}\notin (-1,n) \big\}. 
$$
Here the stopping at $-1$ corresponds to the nonnegativity of the wealth process, whereas the stopping at $n$ is merely for technical convenience. The point in this specific definition of $\cX_{t}^{\rm simp}(\omega)$ is to have a tractable dependence on $\omega$; in this respect, we note that the set $\cH^{\rm simp}$  is independent of $\omega$.
 
\begin{definition}\label{def: super martin abs measures}
  Let $(t,\omega)\in \R_{+}\times\Omega$ and $P\in \fP(\Omega)$. We  introduce the sets
 \[
   \Pa{\zeta^{t,\omega}}{P}=\{Q\in \fP(\Omega):\, Q \domztom P\},
 \]
	\[
	  \cY_{t}(\omega)=\big\{ Q\in\fP(\Omega):\, X\1_{[\![0,\zeta^{t,\omega}-t[\![} \mbox{ is a $Q$-supermartingale $\forall\; X\in\cX_{t}^{\rm simp}(\omega)$}\big\},
	\]
	\[
	\cQ_{t}(\omega,P)=\Pa{\zeta^{t,\omega}}{P}\cap \cY_{t}(\omega),
	\]
	\[
	  \cQ_{t}(\omega)=\bigcup_{P\in \cP_{t}(\omega)} \cQ_{t}(\omega,P).
	\] 
  The elements of $\cQ_{t}(\omega)$ are called  \emph{prior-to-$\zeta$ absolutely continuous supermartingale  measures} given $(t,\omega)$. 
\end{definition}

We observe that the family $\{ \cQ_{t}(\omega):\, t\in\R_{+},\,\omega \in \Omega\}$ is adapted.   Furthermore, we recall from Theorem~\ref{thm: FTAP} that $\cQ_{0}\neq\emptyset$ under $\NAP$.
In the rest of this subsection, we show that the family $\{ \cQ_{t}(\omega)\}$ inherits from $\{ \cP_{t}(\omega)\}$ the properties of Definition~\ref{def : analytic and stable family}.
    
\begin{proposition}\label{pr:QsatisfiesCondA}
  The family $\{ \cQ_{t}(\omega)\}$ satisfies {\rm (A1)--\rm (A3)}.
\end{proposition}
  
The proof is split into the subsequent lemmas.
For ease of reference, we first state the following standard result.
 
\begin{lemma}\label{lem: map omega Q to E is Borel} 
  Let $A$ be a Borel space and let $(a,\omega)\in A\times \Omega \mapsto \xi(a,\omega)\in \R_{+}$ be Borel-measurable. 
 Then,
 $
 (a,R)\in   A\times \fP(\Omega)\mapsto  E^{R}[\xi(a,\cdot)]
 $
 is Borel-measurable.
 \end{lemma} 
 
\begin{proof}
  See, e.g., Step 1 in the proof of \cite[Theorem~2.3]{NutzVanHandel.12}. 
\end{proof}

\begin{lemma}\label{le:superSimple} 
  There exist a countable set $\tilde \cH\subset\cH^{\rm simp}$ and a countable set $\tilde\cT\subset \cT$ of bounded stopping times with the following property:
  
  Given $(t,\omega)\in\R_{+}\times \Omega$ and $Q\in\fP(\Omega)$ such that $S^{t,\omega}$ is $Q$-a.s.\ continuous  prior to $\zeta^{t,\omega}-t$, we have equivalence between
  \begin{enumerate} 
  \item $X\1_{[\![0,\zeta^{t,\omega}-t[\![}$ is a $Q$-supermartingale for all $X\in\cX_{t}^{\rm simp}(\omega)$, 
  \item $X\1_{[\![0,\zeta^{t,\omega}-t[\![}$ is a $Q$-supermartingale for all $X\in\tilde\cX_{t}(\omega)$,
  \item $E^{Q}[X_{\sigma}\1_{\sigma < \zeta^{t,\omega}-t}] \geq E^{Q}[X_{\tau}\1_{\tau < \zeta^{t,\omega}-t}]$ for $X\in\tilde\cX_{t}(\omega)$ and $\sigma\leq \tau$ in $\tilde\cT$,
  \end{enumerate}
  where $\tilde\cX_{t}(\omega)$ is defined like~\eqref{eq:defXsimple} but using only integrands $H\in\tilde \cH$. 
  Moreover, if $S^{t,\omega}\1_{[\![0,\zeta^{t,\omega}-t[\![}$ is a semimartingale under $Q$, the above are equivalent to 
\begin{enumerate}
  \item[(iv)] $X\1_{[\![0,\zeta^{t,\omega}-t[\![}$ is a $Q$-supermartingale for all $X\in\cX_{t}(\omega)$, 
  \end{enumerate}
  where $\cX_{t}(\omega)$ is defined like~\eqref{eq:defXsimple} but using arbitrary predictable integrands.
\end{lemma}

\pagebreak[2] 

\begin{proof}
  For each $s\geq 0$, let $\tilde\cF_{s}$ be a countable algebra generating $\cF_{s}$; cf.\ Lemma~\ref{le:OmegaPolish}. Let $\tilde\cT$ be the set of all stopping times
  $$
    \tau=\sum_{j=1}^{n} t_{j}\1_{A_{j}},
  $$
  where $n\in \N$, $t_{j}\in\Q_{+}$ and $A_{j}\in \tilde\cF_{t_{j}}$. Moreover, let $\tilde \cH\subset\cH^{\rm simp}$ be the set of all processes
$$
 H= \sum_{j=0}^{n} \alpha_{j}\1_{]t_{j}, t_{j+1}]},
$$
where $n\in \N$, $0=t_{0}\leq t_{1}\leq \cdots \leq t_{n}\in\Q_{+}$ and each random variable $\alpha_{j}$ is of the form
$$
 \alpha_{j}= \sum_{i=0}^{n} a_{j}^{i}\1_{A_{j}^{i}}
$$
for some $a_{j}^{i}\in\Q^{d}$ and $A_{j}^{i}\in\tilde\cF_{t_{j}}$.
  
  It is clear that (i)$\Rightarrow$(ii)$\Rightarrow$(iii). To see that (iii) implies (i), fix $Q\in\fP(\Omega)$ and $X\in\cX_{t}^{\rm simp}(\omega)$. We first observe that it suffices to show that
  \begin{enumerate}
  \item[(i')] $E^{Q}[X_{\sigma}\1_{\sigma < \zeta^{t,\omega}-t}] \geq E^{Q}[X_{\tau}\1_{\tau < \zeta^{t,\omega}-t}]$ for all $X\in\cX_{t}^{\rm simp}(\omega)$ and all $\sigma\leq \tau$ in $\tilde\cT$.
  \end{enumerate}  
  Indeed, since $\tilde\cT$ contains all stopping times of the form $\tau=u\1_{A}+v\1_{A^{c}}$ and $\sigma=u$, where $u\leq v\in  \Q_{+}$  and $A\in\tilde\cF_{u}$, it readily follows that (i') implies the supermartingale property of $X\1_{[\![0,\zeta^{t,\omega}-t[\![}$ at rational times, and then the supermartingale property on $\R_{+}$ follows by right-continuity.
  
  To show that (iii) implies (i'), fix $\sigma\leq\tau$ and let $T\in\R_{+}$ be such that $\tau\leq T$. The claim will follow by passing to suitable limits in the inequality
  \begin{equation}\label{eq:ineqLimits}
    E^{Q}[X_{\sigma}\1_{\sigma < \zeta^{t,\omega}-t}] \geq E^{Q}[X_{\tau}\1_{\tau < \zeta^{t,\omega}-t}];
  \end{equation}
  we confine ourselves to a sketch of the proof. Let $X\in\cX_{t}^{\rm simp}(\omega)$ be given and recall that $S^{t,\omega}$ is ($Q$-a.s.) continuous prior to $\zeta^{t,\omega}-t$. 
  
	Using a stopping argument and monotone convergence, we may reduce to the case where $\bar{X}:=X\1_{[\![0,\zeta^{t,\omega}-t[\![}$ is uniformly bounded. Then, using dominated convergence and another stopping argument, we may reduce to the case where $\bar{X}$ is also uniformly bounded away from zero  prior to $\zeta^{t,\omega}-t$. Using standard arguments we can find a sequence $(H^{k})$ of simple predictable integrands with deterministic jump times such that $X^{k}:=1+H^{k}\sint S^{t,\omega} \to X$ uniformly on $[\![0,\zeta^{t,\omega}-t[\![$ in $Q$-probability. Using that $X$ is bounded and bounded away from zero, it follows that 
	$$
	  \bar{X}^{k}:=1+(H^{k}\sint S^{t,\omega})^{\tau^{n}_{H^{k},S^{t,\omega}}}\1_{[\![0,\zeta^{t,\omega}-t[\![} \to \bar{X}
	$$
	uniformly on $[0,T]$ in $Q$-probability, for a sufficiently large $n\in\N$. After an additional approximation, we may obtain the same property with $H^{k}\in\tilde\cH$, and we may show using dominated convergence that the validity of~\eqref{eq:ineqLimits} for each $\bar{X}^{k}$ implies the validity for $\bar{X}$.
	
	If $S^{t,\omega}$ is a semimartingale under $Q$, one shows that (iii) implies (iv) by using similar arguments as well as standard results about stochastic integrals, in particular \cite[Theorems II.21 and IV.2]{Protter.05}.
\end{proof}

 \begin{lemma}\label{lem: A1 holds}
   The family $\{ \cQ_{t}(\omega)\}$ satisfies {\rm (A1)}. 
 \end{lemma}
 
 \begin{proof} 
   Fix $t\geq0$. It suffices to show that the set
 \[
   \Gamma:=\{(\omega,P,Q) :\omega\in \Omega,\; P\in \cP_{t}(\omega),\; Q\in \cQ_{t}(\omega,P)\} \subset \Omega\times\fP(\Omega)\times \fP(\Omega)
 \]
 is analytic. Indeed, once this is established, the graph of $\cQ_{t}(\cdot)$ is analytic as a projection of $\Gamma$; that is, {\rm (A1)} is satisfied.
 
 As a first step, we show that
  \begin{equation}\label{eq:analyticity2}
  \graph(\Pa{\zeta^{t,\cdot}}{\cdot}):=\{(\omega,P,Q ):\, \omega \in \Omega, \;P\in \fP(\Omega),\;Q\in \Pa{\zeta^{t,\omega}}{P}\}\quad\mbox{is Borel}
  \end{equation} 
  and in particular analytic. Indeed, it follows from Lemma~\ref{lem : PXa countable representation} that  
 $$
 \Pa{\zeta^{t,\omega}}{P}=\bigcap_{q\in \Q_{+}} \Pa{\zeta^{t,\omega}}{P,q},
 $$
 where 
 $$
   \Pa{\zeta^{t,\omega}}{P,q} := \big\{Q\in\fP(\Omega):\, Q\ll P \mbox{ on }\cF_{q} \cap \{q<\zeta^{t,\omega}-t\}\big\}.
 $$
 Hence, it suffices to show that 
$$
 \{(\omega, P,Q)\in \Omega\times\fP(\Omega)\times \fP(\Omega):\;\, Q\in \Pa{\zeta^{t,\omega}}{P,q}\}
$$
is Borel  for fixed $q$. Since $\cF_{q}$ is countably generated, cf.\ Lemma~\ref{le:OmegaPolish}, a standard argument (see \cite[Theorem~V.58, p.\,52]{DellacherieMeyer.78} and the subsequent remarks) shows that we can construct a Borel function $D_{q}: \Omega \times  \fP(\Omega)\times \fP(\Omega)\to \R$ such that $D_{q}(\cdot ,Q,P)$ is a version of the Radon--Nikodym derivative of the absolutely continuous part of $Q$ with respect to $P$ on $ \cF_{q}$.  Then, $Q\in \Pa{\zeta^{t,\omega}}{P,q}$ if and only if $E^{P}[D_{q}( Q,P)\1_{q<\zeta^{t,\omega}-t}]=Q\{q<\zeta^{t,\omega}-t\}$.  Using the fact that 
$$
(\omega,P,Q)\mapsto E^{P}[D_{q}(Q,P)\1_{q<\zeta^{t,\omega}-t}]-Q\{q<\zeta^{t,\omega}-t\}
$$ 
is Borel by Lemmas~\ref{lem: map omega Q to E is Borel} and~\ref{le:OmegaPolish}, we conclude that~\eqref{eq:analyticity2} holds.
  
Let $\sigma,\tau\in\cT$ and let $X\in\cX_{t}^{\rm simp}(\omega)$; recall that $X=X^{H}$ is of the form~\eqref{eq:defXsimple}. Then the map
 $$
 (\omega,Q)\in \Omega\times \fP(\Omega)\mapsto \psi^{H,\sigma,\tau}(\omega,Q):= E^{Q}[X_{\tau}\1_{\tau < \zeta^{t,\omega}-t}] - E^{Q}[X_{\sigma}\1_{\sigma < \zeta^{t,\omega}-t}] 
 $$
 is Borel as a consequence of Lemma~\ref{lem: map omega Q to E is Borel}. If $(\omega,Q)$ are such that 
 $S^{t,\omega}$ is $Q$-a.s.\ continuous  prior to $\zeta^{t,\omega}-t$, Lemma~\ref{le:superSimple} shows that
  $$
  Q\in   \cY_{t}(\omega) \quad \mbox{if and only if}\quad \psi^{H,\sigma,\tau}(\omega,Q)\leq0 \;\; \forall   H\in\tilde \cH, \, \sigma\leq\tau\in\tilde\cT.  $$ 
 Using the obvious embeddings of  $\graph(\cP_{t})$ and $\graph(\cY_{t})$ into $\Omega\times\fP(\Omega)\times \fP(\Omega)$, it follows that
\begin{align*}
  \Gamma & = \graph(\cP_{t})  \cap \graph(\Pa{\zeta^{t,\cdot}}{\cdot}) \cap \graph(\cY_{t}) \\
   & = \graph(\cP_{t})  \cap \graph(\Pa{\zeta^{t,\cdot}}{\cdot}) \cap \bigcap_{H\in\tilde \cH, \, \sigma\leq\tau\in\tilde\cT} \{\psi^{H,\sigma,\tau}\leq 0\}.
\end{align*}
Here we have used that if $(\omega,P,Q)$ belong to the first intersection, then $S^{t,\omega}$ is $P$-a.s.\ and hence $Q$-a.s.\ continuous  prior to $\zeta^{t,\omega}-t$; cf.\ Assumption~\ref{ass : Pc(t,omega)}. The above representation shows that $\Gamma$ is analytic as a countable intersection of analytic sets.
\end{proof}

\begin{lemma}
  The family $\{ \cQ_{t}(\omega)\}$ satisfies {\rm (A2)}.
\end{lemma}
 
 \begin{proof}
  For simplicity of notation, we state the proof for $s=0$; the extension to the general case is immediate. Fix $Q\in \cQ_{0}$; then $Q\in \cQ_{0}(P)$ for some $P\in \cP$. We shall show that
\[%begin{equation}\label{eq:A2aim}
  Q^{t,\omega}\in \Pa{\zeta^{t,\omega}}{P^{t,\omega}}\cap \cY_{t}(\omega)  \quad\mbox{for $Q$-a.e.\ $\omega\in\{\zeta>t\}$};
\]%end{equation}
 this will imply the lemma because $P^{t,\omega}\in\cP_{t}(\omega)$ holds for $P$-a.e.\ $\omega\in\Omega$, cf.\ Assumption~\ref{ass : Pc(t,omega)}, and thus for $Q$-a.e.\ $\omega\in\{\zeta>t\}$ as $Q\in \cQ_{0}(P)$.
 
 Let $Y$ be the prior-to-$\zeta$ density process of $Q$ with respect to $P$ (see Remark~\ref{rem : density process constructed from Q} for details on this notion) and set
 $$
 \tilde Y=\1_{[0,t)}+(Y/Y_{t})\1_{[t,\infty)},
 $$
 where we use the convention $0/0=0$. We first establish that given $s\geq0$, we have $Q^{t,\omega}\ll P^{t,\omega}$ on $\cF_{s}\cap \{\zeta^{t,\omega}-t>s\}$ and in fact
 $$
   dQ^{t,\omega}=\tilde Y_{s}^{t,\omega} dP^{t,\omega}\quad\mbox{on}\quad \cF_{s}\cap \{\zeta^{t,\omega}-t>s\}
 $$
 for $Q$-a.e.\ $\omega\in\{\zeta>t\}$. Indeed, let $g\geq0$ be an $\cF_{s}$-measurable random variable; then
 there exists an $\cF_{s+t}$-measurable random variable $\bar g$ such that $\bar g^{t,\omega}=g$. Recalling~\eqref{eq:conditioning}, we have for $Q$-a.e.\ $\omega\in\{\zeta>t\}$ that
 \begin{align*}
 E^{Q^{t,\omega}}[g\1_{\zeta^{t,\omega}-t>s}]
 &=E^{Q}[\bar g\1_{\zeta>s+t}|\cF_{t}](\omega)\\
 &=E^{P}[(Y_{s+t}/Y_{t})\bar g\1_{\zeta>s+t}|\cF_{t}](\omega)\\
 &=E^{P}[\tilde Y_{s+t}\bar g\1_{\zeta>s+t}|\cF_{t}](\omega)\\
 &=E^{P^{t,\omega}}[\tilde Y^{t,\omega}_{s}g\1_{\zeta^{t,\omega}-t>s}].  
 \end{align*}
 We have shown in particular that $Q^{t,\omega}\ll P^{t,\omega}$ on $\cF_{s}\cap \{\zeta^{t,\omega}-t>s\}$ for all $s\in\Q_{+}$ holds for $Q$-a.e.\ $\omega\in\{\zeta>t\}$, which by Lemma~\ref{lem : PXa countable representation} implies that
 \[%begin{equation}\label{eq:A2proofAim2}   
   Q^{t,\omega}\in \Pa{\zeta^{t,\omega}}{P^{t,\omega}}   \quad \mbox{for $Q$-a.e.\ $\omega\in \{\zeta>t\}$.}
 \]%end{equation}
 It remains to prove that 
 \begin{equation}\label{eq:A2proofAim1}
    Q^{t,\omega}\in \cY_{t}(\omega)\quad \mbox{for $Q$-a.e.\ $\omega\in \{\zeta>t\}$.}
 \end{equation}
 Let $X\in \cX_{t}^{\rm simp}(\omega)$, then we observe that $X=\bar{X}^{t,\omega}$ for some $\bar{X}\in \cX^{\rm simp}_{0}$.  Moreover, let $\sigma\in\cT$   be bounded,  then $\sigma=\bar\sigma^{t,\omega}-t$ for some bounded $\bar\sigma\in\cT$ satisfying $\bar\sigma\geq t$ (both $\bar{X}$ and $\bar\sigma$ do not depend on $\omega$). We have $X_{\sigma}=(\bar{X}^{t,\omega})_{\bar\sigma^{t,\omega}-t}=(\bar{X}_{\bar\sigma})^{t,\omega}$ (where $\bar{X}_{\bar\sigma}$ is considered as a random variable) and thus
$$
  E^{Q^{t,\omega}}[X_{\sigma}\1_{\zeta^{t,\omega}-t>\sigma}]=E^{Q^{t,\omega}}[(\bar{X}_{\bar\sigma})^{t,\omega}\1_{\zeta^{t,\omega}>\bar\sigma^{t,\omega}}] = E^{Q}[\bar{X}_{\bar\sigma} \1_{\zeta>\bar\sigma}|\cF_{  t }](\omega)
$$ 
for $Q$-a.e.\ $\omega\in \{\zeta>t\}$.
If $\tau\geq\sigma \in\cT$ is bounded and $\bar\tau\geq \bar\sigma $ has the obvious meaning, we deduce from the supermartingale property of $Q\in\cY_{0}$ that
\begin{align*}
  E^{Q^{t,\omega}}[X_{\sigma}\1_{\zeta^{t,\omega}-t>\sigma}] 
  &= E^{Q}[\bar{X}_{\bar\sigma} \1_{\zeta>\bar\sigma}|\cF_{t}](\omega) \\
  & \geq E^{Q}[\bar{X}_{\bar\tau} \1_{\zeta>\bar\tau}|\cF_{t}](\omega) \\
  &= E^{Q^{t,\omega}}[X_{\tau}\1_{\zeta^{t,\omega}-t>\tau}]
\end{align*}
for $Q$-a.e.\ $\omega\in \{\zeta>t\}$. Now Lemma~\ref{le:superSimple} implies \eqref{eq:A2proofAim1} and the proof is complete.
%
% 
%  We know that $YX$, and therefore $\tilde YX$, is a $P$-supermartingale on $[t,\infty)$. Hence, it follows from~\eqref{eq:conditioning} that 
% \[
% E^{Q^{t,\omega}}[g\1_{\zeta^{t,\omega}>T}] = E^{P^{t,\omega}}[\tilde Y_{T-t}^{t,\omega} X_{T-t}^{t,\omega}\1_{\zeta^{t,\omega}>T}] 
% = E^{P}[\tilde Y_{T}X_{T}| \cF_{t}](\omega) 
% \le 1\;
% \]
% for $P$-a.e.~$\omega\in \Omega$.  Since $Q\ll_{\zeta}P$, the above also holds for $Q$-a.e.\ $\omega\in \{\zeta>t\}$. Here the exceptional nullset depends on $g$, but Lemma~\ref{lem : reduction supermatingale condition - simple to countable for Z} shows that the condition $Q^{t,\omega}\in \cY^{T}_{t}(\omega)$ may be verified using only countably many $g\in\cC^{T}_{0}$ and so we may nevertheless conclude that~\eqref{eq:A2proofAim1} holds.
% As~\eqref{eq:A2aim} is the combination of~\eqref{eq:A2proofAim1} and~\eqref{eq:A2proofAim2}, the proof is %complete.
 \end{proof}  
   
\begin{lemma}  
  The family $\{ \cQ_{t}(\omega)\}$ satisfies {\rm (A3)}.
\end{lemma}
 
 \begin{proof} 
  Again, we state the argument for the case $s=0$. Let $Q\in \cQ_{0}$; then $Q\in \cQ_{0}(P)$ for some $P\in \cP  =\cP_{0}$.  Moreover, let $t\ge 0$ and let $\nu$ be an $\cF_{t}$-measurable kernel such that  $\nu(\omega)\in \cQ_{t}(  \omega)$  for $Q$-a.e.\ $\omega \in \{\zeta>t\}$. Using Assumption~\ref{ass : Pc(t,omega)} and the measurability results established in the proof of Lemma~\ref{lem: A1 holds}, it follows that the set 
 $$
  \{(\omega,P',Q') :\omega\in \Omega,\; P'\in \cP_{t}(  \omega),\; Q' =\nu(\omega),\;Q'\in  \cQ_{t}(  \omega,P')\}  
 $$
 is analytic.   Let  $\cF^{*}_{t}$ be the universal completion of $\cF_t$.  Applying the measurable selection theorem, cf.\ \cite[Proposition~7.49]{BertsekasShreve.78},  we can find an $\cF^{*}_{t}$-measurable   kernel $\mu'$ such that  $\mu'(\omega)\in \cP_{t}(\omega)$ and $\nu(\omega)\in \cQ_{t}(  \omega,\mu'(\omega))$ for all $\omega\in\{\zeta>t\}$ outside the $\cF^{*}_{t}$-measurable $Q$-nullset 
  $$
    N':=\{\nu\notin\cQ_{t}\}\cap \{\zeta>t\}
  $$
 and, e.g., $\mu'(\omega)=P^{t,\omega}$ for $\omega\in N'$. We can then find an $\cF_{t}$-measurable kernel $\mu$ and a $P$-nullset $N$ such that $\mu(\omega)=\mu'(\omega)$ for all $\omega\notin N$; cf.\  \cite[Lemma~7.27]{BertsekasShreve.78}.  Using Assumption~\ref{ass : Pc(t,omega)} and $Q\ll_{\zeta}P$, we have
 \begin{align}
   \mu(\omega)&\in \cP_{t}(\omega) &\mbox{for $P$-a.e.\ $\omega\in \{\zeta>t\}$;}\nonumber\\
   \nu(\omega)&\in \cQ_{t}(  \omega,  \mu(\omega))  &\mbox{for $Q$-a.e.\ $\omega\in \{\zeta>t\}$.}\label{eq: nu in Q of mu}
 \end{align}  
 By Assumption  \ref{ass : Pc(t,omega)}, the measure 
 $$
\bar P(A):=\iint (\1_{A})^{t,\omega}(\omega')\,\mu^{P}(d\omega';\omega)\,P(d\omega),\quad A\in \cF
 $$
 is an element of $\cP$; cf.\ Definition~\ref{def : analytic and stable family} for the notation. Set
 $$
\bar Q(A):=\iint  (\1_{A})^{t,\omega}(\omega')\,\nu^{Q}(d\omega';\omega)\,Q(d\omega),\quad A\in \cF.
 $$
 Next, we show that $\bar Q \ll_{\zeta} \bar P$; i.e., that
  \[
    \bar Q \ll \bar P \quad \mbox{on}\quad \cF_{s}\cap\{s<\zeta\}, \quad s\geq0.
  \]
  This is clear for $s\leq t$ since $\bar Q = Q \ll_{\zeta} P = \bar P$ on $\cF_{t}$. Let $s>t$ and let $A\in\cF_{s}$ be such that $\bar P(A\cap \{s<\zeta\})=0$. Then 
  $$
    \mu(\omega)\{(A\cap \{s<\zeta\})^{t,\omega}\}=\bar P^{t,\omega}\{(A\cap \{s<\zeta\})^{t,\omega}\}=0
  $$
  and thus 
  $$
    \bar Q^{t,\omega}\{(A\cap \{s<\zeta\})^{t,\omega}\}=\nu(\omega)\{(A\cap \{s<\zeta\})^{t,\omega}\}=0
  $$
  for $Q$-a.e.\ $\omega\in \{\zeta>t\}$,   by \eqref{eq: nu in Q of mu}.  It follows that 
  $$
    \bar Q(A\cap \{s<\zeta\})=E^{Q}\big[E^{\bar Q}[\1_{A\cap \{s<\zeta\}}|\cF_{t}]\big]=0
  $$
  as desired.
  
  To see that $\bar Q\in\cY_{0}$, let $X\in\tilde\cX_{0}$ (recall the notation from Lemma~\ref{le:superSimple});  then $X\1_{[\![0,\zeta[\![}$ is a $Q$-supermartingale. Moreover, noting that $X^{t,\omega}$ is an element of the scaled space $X_{t}(\omega)\cX_{t}^{\rm simp}(\omega)$, we have that $X^{t,\omega}\1_{[\![0,\zeta^{t,\omega}-t[\![}$ is a $\nu(\omega)$-supermartingale for all $\omega$ such that $\nu(\omega)\in\cQ_{t}(\omega)$. Using Fubini's theorem, it then follows that $X\1_{[\![0,\zeta[\![}$ is a $\bar Q$-supermartingale as desired.
  
  We have shown that  $\bar Q\in \Pa{\zeta }{\bar P}\cap \cY_{0}\subset \cQ_{0}$ and the proof is complete.
 \end{proof}

%%%%%%%%%%%%%%%%%%%%%%%%%%%%%%%%%

\section{Superhedging Duality}\label{sec: superhedging duality}
 
In this section, we provide a superhedging duality and the existence of an optimal strategy. To this end, we require an enlargement of the set of admissible strategies, allowing for continuous trading.
 We first introduce the filtration $\G= (\cG_t)_{t\geq 0}$, where
\begin{equation*}
  \cG_t:= \cF^{*}_{t}\vee \mathcal{N}^{\cP};
\end{equation*}
here $\cF^{*}_{t}$ is the universal completion of $\cF_t$ and $\mathcal{N}^{\cP}$ is the collection of sets which are $(\cF,P)$-null for all $P\in\cP$. Moreover, Assumption~\ref{ass : Pc(t,omega)} is in force throughout this section.

Let $\NAP$ hold, then Corollary~\ref{co:FTAP} implies the $(\G,P)$-semimartingale property of $S$ for each $P\in \cP$. Therefore, we may introduce the class  ${\cal L}(\cP)$  of all predictable processes on $(\Omega, \bG)$ that are $S$-integrable under every $P \in \Prob$.
Given $H\in {\cal L}(\cP)$ and $P\in\cP$,  we can construct the usual stochastic integral $H\sint S$ under $P$ (the dependence on $P$ is suppressed in the notation---but see also~\cite{Nutz.11int}).   
For $x \in \Real_+$, we denote by $\Hcal (x)$ the collection of all $H \in {\cal L}(\cP)$ such that  $x + H \sint S$ remains $P$-a.s.\ nonnegative for all $P \in \Prob$. 

To be consistent with the classical literature, the following superhedging theorem is stated with the set
$$
\cQ:=\bigcup_{P\in \cP} \cQ^{P}
$$ 
of prior-to-$\zeta$ local martingale measures; cf.\ Definition~\ref{defn:ELMM}. The subsequent  Lemma~\ref{le:sameSupremum} provides an equivalent version with the set $\cQ_{0}$ of supermartingale measures. 

\begin{theorem}\label{th:duality}
  Let $\NAP$ hold, let   $T\in \R_{+}$ and  let
  $f:\Omega\to [0,\infty]$ be an upper semianalytic
  \footnote{The definition of an upper semianalytic function is recalled in Section~\ref{sec:measureTheory} of the Appendix. In particular, any Borel function is upper semianalytic.}, $\cG_T$-measurable function with  $\sup_{Q\in\cQ} E^Q[f\1_{\zeta>T}]<\infty$.
  Then
  \begin{multline*}
    \sup_{Q\in\cQ} E^Q[f\1_{\zeta>T}] \\
    =\min\big\{x:\,\exists\, H\in \cH(x)\mbox{ with } x+ (H\sint S)_{T} \geq f\;P\as\mbox{ for all }P\in\cP\big\}.
  \end{multline*}
\end{theorem}

In order to prove this theorem, we first show that $\cQ$ can equivalently be replaced by $\cQ_0$ in its statement.

\begin{lemma}\label{le:sameSupremum}   
  Let $\NAP$ hold, let  $T\in \R_{+}$ and let
  $f:\Omega\to [0,\infty]$ be a $\cG_{  T }$-measurable function. Then
  \[
    \sup_{Q\in\cQ} E^Q[f\1_{\zeta>T}] = \sup_{Q\in\cQ_{0}} E^Q[f\1_{\zeta>T}].
  \]
\end{lemma}

\begin{proof} Since $\cQ \subseteq \cQ_{0}$, we only have one non-trivial inequality to prove. Fix $Q_0 \in \cQ_{0}$, and let $P \in \Prob$ be such that $Q_0\ll_{\zeta} P$. By Remark \ref{rem : density process constructed from Q} in the Appendix, one can construct a c\`adl\`ag adapted process $Y^0\geq0$ which is the prior-to-$\zeta$ density of $Q_{0}$ with respect to $P$. Then, the same arguments as in \cite[Proposition 3.2]{LarsenZitkovic.07} show that one may write $Y^0 = Y D$, where $D$ is an $\bF_+$-predictable nonincreasing process with $D_{0}=1$ and $Y$ is a $P$-a.s.\ strictly positive c\`adl\`ag $(\bF_+, P)$-local martingale such that $Y S$ is also an $(\bF_+, P)$-local martingale. Applying Theorem~\ref{thm: foelmeasure_1} from the Appendix, we construct $Q \simz P$ whose  prior-to-$\zeta$ density with respect to $P$ is $Y$.
Clearly, 
$$
 E^{Q}[f\1_{\zeta>T}]=E^{P}[Y_{T}f] \ge E^{P}[Y^0_{T}f] = E^{Q_{0}}[f\1_{\zeta>T}]
$$ 
since $f\ge 0$. It remains to show that $Q \in  \cQ^{P}$, which follows in a straightforward way from the fact that $Y S$ is an $(\bF_+, P)$-local martingale and that $\zeta$ is foretellable under~$Q$; see Definition~\ref{defn: foretellable} and Theorem~\ref{thm: foelmeasure_1} in the Appendix.
\end{proof}

%\subsection{Proof of Theorem~\ref{th:duality}}\label{subsec: proof of thm superhedging}

The remainder of this section is devoted to the proof of Theorem~\ref{th:duality}. 
In the course of this proof, $T>0$ is fixed and $f$   satisfies the assumptions stated in the theorem. We will use Lemma~\ref{le:sameSupremum} without further mention. To simplify the notation, we may assume that 
$$
  S=S\1_{[\![0,\zeta[\![}
$$
and moreover we set 
$$
  g:=f\1_{\zeta>T};
$$
note that $g$ is upper semianalytic like $f$.

We begin by proving the easy inequality of the theorem. Let $x\in\R$ and suppose there exists $H\in\cH(x)$ such that $x+ H\sint S_{T}\geq g$ $P$-a.s.\ for all $P\in\cP$. Fix $Q\in\cQ$; then there exists $P\in\cP$ such that $Q\sim_{\zeta}P$. Remark~\ref{rem:foretell_and_pred} from the Appendix shows that $\zeta$ is a predictable stopping time in the $Q$-augmentation $\G^{Q}_{+}$ of $\G_{+}$.  It follows that $H':=H \1_{[\![0,\zeta[\![}$ is predictable in $\G^{Q}_{+}$, and thus $x+H'\sint S$ is a nonnegative local martingale under $Q$; in particular, a $Q$-supermartingale. Using that $g=0$ on $\{\zeta\leq T\}$, we see that $x+ H'\sint S_{T}\geq g$ $Q$-a.s., and now taking expectations yields $x\geq E^Q[g]$. Since $Q\in\cQ$ was arbitrary, the inequality ``$\geq$'' of the theorem follows.

To complete the proof of the theorem, we shall construct in the remainder of this section a strategy $H$ satisfying
\begin{equation}\label{eq:aimH}
  \sup_{Q\in\cQ} E^{Q}[g] + H\sint S_{T} \geq g \quad P \mbox{-a.s.} \quad \mbox{for all} \quad P \in \cP.
\end{equation} 
Given $t\geq0$ and an upper semianalytic function $h\geq 0$ on $\Omega$, we define 
\[
  \cE_t(h)(\omega):=\sup_{Q\in\cQ_{t}(\omega)} E^Q[h^{t,\omega}],\quad\omega\in\Omega.
\] 
Moreover, we denote $\F^{*}=(\cF^{*}_{t})_{t\in \R_{+}}$. 

\begin{lemma}\label{le:Esupermart}
  The process $\{\cE_{t}(g)\}_{t\in[0,T]}$ is a $(Q,\F^{*})$-supermartingale for all $Q\in\cQ_{0}$, and in particular for all $Q\in\cQ$.
\end{lemma}

\begin{proof}
  Let $s\leq t$. In view of Proposition~\ref{pr:QsatisfiesCondA} and Lemma~\ref{le:OmegaPolish}, we may adapt the proof of \cite[Theorem~2.3]{NutzVanHandel.12} to establish that $\cE_t(g)$ is $\cF_t^*$-measurable and upper semianalytic, that 
   \[     \cE_s(g\1_{\zeta>t})(\omega) = \cE_s(\cE_t(g)\1_{\zeta>t})(\omega)\quad\mbox{for all}\quad \omega\in\Omega,
   \]
   and that
   \[
     \cE_s(g\1_{\zeta>t}) = \mathop{\esssup^Q}_{Q'\in \cQ^{Q}_{s}} E^{Q'}[\cE_t(g)\1_{\zeta>t}|\cF_s]\quad Q\as\quad\mbox{for all}\quad Q\in\cQ,
   \] 
  where $\cQ^{Q}_{s}=\{Q'\in \cQ:\, Q'=Q \mbox{ on } \cF_s\}$. Since $\{\zeta>T\} \subseteq \{\zeta>t\}$ for $t\leq T$, we have $g\1_{\zeta>t}=f\1_{\zeta>T}\1_{\zeta>t}=f\1_{\zeta>T}=g$. Hence, the above simplifies to
     \begin{equation}\label{eq:DPP}
     \cE_s(g)  = \cE_s(\cE_t(g)),\quad s\leq t\leq T
   \end{equation}
   and
  \begin{equation}\label{eq:esssupDPP}
     \cE_s(g) = \mathop{\esssup^Q}_{Q'\in \cQ^{Q}_{s}} E^{Q'}[\cE_t(g)|\cF_s]\quad Q\as\quad\mbox{for all}\quad Q\in\cQ,\quad s\leq t\leq T.
  \end{equation}
  Our assumption that $\cE_{0}(g)<\infty$ and~\eqref{eq:DPP} applied with $s=0$ yield that 
  $\sup_{Q\in\cQ} E^{Q}[\cE_{t}(g)]<\infty$ for all $t$; in particular, $\cE_{t}(g)$ is integrable under all~$Q\in\cQ$.  Moreover, \eqref{eq:esssupDPP} yields that
  \[
     \cE_s(g) \geq E^{Q}[\cE_t(g)|\cF_s]=E^{Q}[\cE_t(g)|\cF^{*}_s]\quad Q\as\quad\mbox{for all}\quad \!Q\in\cQ,\quad s\leq t\leq T,
  \]
  which is the desired supermartingale property.
\end{proof}

\begin{lemma}\label{le:Zworks}
  Define
  \[
    Z'_{t}:= \limsup_{r\downarrow t,\, r\in \Q} \cE_{r}(g)\quad\mbox{for}\quad t<T \quad\mbox{and}\quad Z'_{T}:= \cE_{T}(g),
  \]
  let $N$ be the set of all $\omega\in\Omega$ such that $Z'(\omega)$ is not c\`adl\`ag, and
  \[
    Z:=Z'\1_{N^{c}}.
  \]
  Then $(Z_{t})_{t\in[0,T]}$ is a c\`adl\`ag, $\G_{+}$-adapted process which is a $Q$-supermartingale for all $Q\in\cQ$. Moreover,
  \begin{equation}\label{eq:aimZ}
    Z_{0} \leq \sup_{Q\in\cQ} E^{Q}[g] \quad \mbox{and} \quad Z_{T}=g \quad P\mbox{-a.s.} \quad \mbox{for all} \quad P \in \cP.
  \end{equation}
\end{lemma}

\begin{proof}
  Recall Lemma~\ref{le:Esupermart}. The modification theorem for supermartingales \cite[Theorem~VI.2]{DellacherieMeyer.82} yields that $N$ is $\cQ$-polar, the limit superior in its definition is actually a limit outside a $\cQ$-polar set, and moreover that $Z'$ is a $(\G_+,Q)$-supermartingale for all $Q\in\cQ$. 

	To see that $N\in\cN^{\cP}$, we fix an arbitrary $P\in\cP$ and show that $N$ is $P$-null. Indeed, we may decompose $N$ as
	$$
	  N = (N \cap \{\zeta \leq T\}) \cup (N \cap \{\zeta > T\}).
	$$
	The first set is $P$-null because $\{\zeta <\infty\}$ was assumed to be $\cP$-polar. We know that there exists $Q\in\cQ$ such that $P\sim_{\zeta} Q$. Since $N$ is $Q$-null relative to $\cF^{*}_{T}$, there exists an $\cF_{T}$-measurable $Q$-nullset $N^{Q}$ such that $N\subseteq N^{Q}$. Now $P\sim_{\zeta} Q$ implies that $  N^{Q}  \cap \{\zeta > T\}$ is $P$-null, and then so is $N \cap \{\zeta > T\}$. As a result, we have $N\in\cN^{\cP}$ and in particular $N\in\cG_{0}$. This implies that 
	$Z:=Z'\1_{N^{c}}$ is still a $(\G_+,Q)$-supermartingale for all $Q\in\cQ$, while in addition all paths of $Z$ are c\`adl\`ag. 
	Moreover, for any $P\in\cP$, it follows from $\cG_{T}=\cF_{T}$ $P$-a.s.\ and \eqref{eq:esssupDPP} that $Z_{T}=Z'_{T}=\cE_{T}(g)=g$ $P$-a.s. 
	
	It remains to show the first part of~\eqref{eq:aimZ}. Since $Z_0$ is $\cG_{0+}$-measurable,  $\cG_{0+}$ is equal to $\cF_{0+}$ up to $P$-nullsets for any $P\in\cP$, and any $P\in\cP$ is dominated on $\cF_{0+}$ by some $Q\in\cQ$, it suffices to show that
	$$
	  Z_0 \leq \sup_{Q'\in\cQ} E^{Q'}[g]\equiv \cE_0(g) \quad Q\mbox{-a.s.}
	$$
	for all $Q\in\cQ$. The proof of this fact is similar to the proof of \cite[Inequality~(3.3)]{Nutz.14}. Namely, it follows from Lemma~\ref{le:Esupermart} and the construction of $Z$ that 
	$$
	 \sup_{Q'\in\cQ}E^{Q'}[Z_0]\leq \sup_{Q'\in\cQ} E^{Q'}[g].
	$$
	Then, one shows that $\sup_{Q'\in\cQ}E^{Q'}[Z_0]$ dominates the $Q$-essential supremum of $Z_{0}$ for any $Q\in\cQ$ by verifying that $\cQ$ is stable under $\cF_{0+}$-measurable, equivalent changes of measure---see Theorem~\ref{thm: foelmeasure_1}.  We omit the details.
\end{proof}

\begin{lemma}\label{le:classicalOptDecomp}
  Let $Q\in\cQ$. Then there exists a $\G_{+}^{Q}$-predictable process $H^{Q}$ which is $S$-integrable under $Q$ such that 
  $$
    Z-  H^{Q}  \sint S\quad \mbox{is nonincreasing $Q$-a.s.\ on $[\![0,\zeta[\![\cap [\![0,T]\!]$.}
  $$
\end{lemma}

\begin{proof}
  Let $\sigma_{n}$ be an announcing sequence for $\zeta$ associated with $Q$ and set $\tau_{n}:=\sigma_{n}\wedge T$. Let $Q'$ be a probability on $\cF_{T}$ which is equivalent to $Q$ and such that $S^{\tau_{n}}$ is a $Q'$-local martingale; we show that $Z^{\tau_{n}}$ is a $Q'$-supermartingale. Indeed, let $Y'=(Y'_{t})_{t\in[0,T]}$ be the density process of $Q'$ with respect to $Q$ and the filtration $\G_{+}^{Q}$, a strictly positive $Q$-martingale with unit expectation. Define
  $$
    Y''_{t} := Y'_{t\wedge \tau_{n}},\quad t\geq 0;
  $$
  then $Y''$ is the density process of a probability $Q''$ with respect to $Q$ and it is elementary to verify that $Q''\in\cQ$. Thus, $Z$ is a $Q''$-supermartingale by Lemma~\ref{le:Zworks}. As $Q''=Q'$ on $\cG_{\tau_{n}+}$, it follows that $Z^{\tau_{n}}$ is a $Q'$-supermartingale as desired. As a result, we may apply the classical optional decomposition theorem (see \cite{FollmerKabanov.98}) to obtain an integrand $H^{Q,n}$ such that
  $$
    Z^{\tau_{n}}- H^{Q,n}\sint S^{\tau_{n}}\quad \mbox{is nonincreasing $Q$-a.s.}
  $$
  The result follows by a passage to the limit $n\to\infty$.
\end{proof}

\begin{proof}[End of the Proof of Theorem~\ref{th:duality}]
	We can now construct $H$ as in \eqref{eq:aimH} by arguments similar to the proof of \cite[Theorem~2.4]{Nutz.14}. To this end, we recall that $S=S\1_{[\![0,\zeta[\![}$. Moreover, as we will be working in the filtration $\G$ and $\cN^{\cP}\subset\cG_{0}$, we may assume without loss of generality that all paths of~$S$ are continuous prior to~$\zeta$. 
	
The $(d+1)$-dimensional process $(S,Z)$ is essentially a $\G_{+}$-semimartingale under all $Q\in\cQ$; that is, modulo the fact that $S$ may fail to have a left limit at $\zeta$. Following the construction of \cite[Proposition~6.6]{NeufeldNutz.13a}\footnote{That proposition does not use the separability assumptions on the filtration that are imposed for the main results of \cite{NeufeldNutz.13a}.}, there exists a $\G_{+}$-predictable (and hence $\G$-predictable) process $C^{(S,Z)}$ with values in $\S^{d+1}_{+}$ (the set of nonnegative definite symmetric matrices), having $\cQ$-q.s.\ continuous and nondecreasing paths prior to $\zeta$, and which coincides $Q$-a.s.\ with $\br{(S,Z)^{c}}^{Q}$ under each $Q\in\cQ$,  prior to $\zeta$. Here $\br{(S,Z)^{c}}^{Q}$ denotes the usual second characteristic of $(S,Z)$ under $Q$; i.e., the quadratic covariation process of the continuous local martingale part of $(S,Z)$.
	
	Let $C^{S}$ be the $d\times d$ submatrix corresponding to $S$ and let $C^{SZ}$ be the $d$-dimensional vector corresponding to the quadratic covariation of $S$ and $Z$. Let $A_{t}:=\tr C^{S}_{t}$ be the trace of $C^{S}$; then, prior to $\zeta$, $C^{S}\ll A$ $\cQ$-q.s.\ and $C^{SZ}\ll A$ $\cQ$-q.s.\ (i.e., absolute continuity holds outside a polar set). Thus, we have $dC^{S} = c^{S}dA$ $\cQ$-q.s.\ and $dC^{SZ} = c^{SZ}dA$ $\cQ$-q.s.\ for the derivatives defined by
	\[
		c^{S}_{t}:= \tilde{c}^{S}_{t} \1_{\{\tilde{c}^{S}_{t}\in \S^{d}_{+}\}}, 
		\quad \tilde{c}^{S}_{t} := \limsup_{n\to\infty} \frac{C^{S}_{t} - C^{S}_{(t-1/n)\vee 0}} {A_{t} - A_{(t-1/n)\vee 0}}
  \]
  and
	\[
		c^{SZ}_{t}:= \tilde{c}^{SZ}_{t} \1_{\{\tilde{c}^{SZ}_{t}\in \R^{d}\}}, 
		\quad \tilde{c}^{SZ}_{t} := \limsup_{n\to\infty} \frac{C^{SZ}_{t} - C^{SZ}_{(t-1/n)\vee 0}} {A_{t} - A_{(t-1/n)\vee 0}},
  \]  
  where all operations are componentwise and $0/0:=0$.
  Let $(c^{S})^{\oplus}$ be the Moore--Penrose pseudoinverse of $c^{S}$ and define the $\G$-predictable process
	\[
	  H:= \begin{cases}
	    c^{SZ} (c^{S})^{\oplus} &\mbox{on }[\![0,\zeta[\![\cap [\![0,T]\!], \\
	    0 & \mbox{otherwise};
	  \end{cases}
	\]	
	we show that $H$ satisfies~\eqref{eq:aimH}.
	
	Fix $Q\in\cQ$. By Lemma~\ref{le:classicalOptDecomp}, there exist an $S$-integrable process $H^{Q}$ and a nondecreasing process $K^{Q}$ such that 
	\begin{equation}\label{eq:usualOptDecomp}
	  Z=Z_{0}+ H^{Q}\sint S - K^{Q}\quad \mbox{$Q$-a.s.\ on $[\![0,\zeta[\![\cap [\![0,T]\!]$.}
	\end{equation}
  It follows that	
  \begin{equation*}
	  d\br{S,Z}=H^{Q} d\br{S} \quad Q\as,
	\end{equation*}
	or equivalently
	$$
	  c^{SZ}=H^{Q} c^{S} \quad Q\times dA\mbox{-a.e.}
	$$
	By It\^o's isometry, this implies that $H$ is $S$-integrable under $Q$ and
	$$
	  H\sint S = H^{Q}\sint S \quad \mbox{$Q$-a.s.\ on $[\![0,\zeta[\![\cap [\![0,T]\!]$.}
	$$
	Now~\eqref{eq:usualOptDecomp} implies that
	$$
    Z-Z_{0}- H\sint S \quad \mbox{is nonincreasing and nonpositive $Q$-a.s.\ on $[\![0,\zeta[\![\cap [\![0,T]\!]$.}
  $$
  Noting that
  $$
    Z_{t}\1_{\zeta \leq t} = \cE_{t+}(f\1_{\zeta > T})\1_{\zeta \leq t} = \cE_{t+}(f\1_{\zeta > T}\1_{\zeta \leq t})=\cE_{t+}(0)=0\quad Q\as,
  $$
  we see that $Z=0$ on $[\![0,T]\!] \setminus [\![0,\zeta[\![$. Since $H$ also vanishes on that set, we conclude that
	$$
    Z-Z_{0}- H\sint S\quad \mbox{is nonincreasing and nonpositive $Q$-a.s.\ on $[\![0,T]\!]$.}
  $$  
  In particular, $Z_{0}+H\sint S \geq0$ $Q$-a.s. As $Q\in\cQ$ was arbitrary, it easily follows that 
  $Z_{0}+H\sint S \geq0$ $P$-a.s.\ and that
	$$
    Z-Z_{0}- H\sint S\quad \mbox{is nonincreasing $P$-a.s.\ on $[\![0,T]\!]$}
  $$  
  for all $P\in\cP$. Thus, we have
  \[
    \sup_{Q\in\cQ} E^{Q}[g] + H\sint S_{T} \geq Z_{0} +H\sint S_{T} \geq Z_{T} = g \quad P\mbox{-a.s.} \quad \mbox{for all} \quad P \in \cP
  \]
  and $H\in\cH(x)$ for $x=\sup_{Q\in\cQ} E^{Q}[g]$. This completes the proof of~\eqref{eq:aimH} and thus of Theorem~\ref{th:duality}.
\end{proof}

\appendix
\section{Appendix}\label{sec : Appendix}

\subsection{Notions from Measure Theory}\label{sec:measureTheory}

Given a measurable space $(\Omega,\cA)$, let $\fP(\Omega)$ the set of all probability measures on $\cA$. The \emph{universal completion} of $\cA$ is the $\sigma$-field $\cap_{P\in\fP(\Omega)} \cA^P$, where $\cA^P$ denotes the $P$-completion of $\cA$. When $\Omega$ is a topological space with Borel $\sigma$-field $\cB(\Omega)$, we endow $\fP(\Omega)$ with the topology of weak convergence. Suppose that $\Omega$ is Polish, then $\fP(\Omega)$ is Polish as well.  A subset $A\subset\Omega$ is called \emph{analytic} if it is the image of a Borel subset of another Polish space under a Borel-measurable mapping. Analytic sets are stable under countable union and intersection, and under forward and inverse images of Borel functions. Any Borel set is analytic, and any analytic set is universally measurable. A function $f: \Omega\to [-\infty,\infty]$ is \emph{upper semianalytic} if $\{f\geq c\}$ is analytic for every $c\in\R$. In particular, any Borel function is upper semianalytic. We refer to \cite[Chapter~7]{BertsekasShreve.78} for these results and further background.

\subsection{F\"ollmer's Exit Measure}\label{sec: follmer measure}

  Important references on F\"ollmer's exit measure are \cite{Foellmer.72} and \cite{Meyer.72}; see also \cite{PerkowskiRuf.13} and the references therein for recent developments. The first result of this section provides an alternative, seemingly stronger characterization of the notion of prior-to-$\zeta$ absolute continuity---compare with Definition~\ref{def : prior to zeta abs cont and equiv}.  
   
\begin{lemma}\label{lem : PXa countable representation} 
  Let $\xi$ be a random time and $P,Q\in\fP(\Omega)$. Then
  \begin{equation}\label{eq: prior to zeta equiv}
   P(A\cap \{\tau<\xi\})=0\quad\Rightarrow \quad  Q(A\cap \{\tau<\xi\})=0\quad\forall\; \tau\in \cT_{+},\;A\in \cF_{\tau+}
 \end{equation}
 holds if and only if 
  \begin{equation}\label{eq: prior to zeta equiv coutable}
   P(A\cap \{q<\xi\})=0\quad\Rightarrow \quad   Q(A\cap \{q<\xi\})=0\quad\forall\; q\in \Q_{+},\;A\in \cF_{q}. 
  \end{equation}
\end{lemma}
 
\begin{proof} 
 It is clear that~\eqref{eq: prior to zeta equiv} implies~\eqref{eq: prior to zeta equiv coutable}. For the converse, we first note that it suffices to check~\eqref{eq: prior to zeta equiv} for $\F$-stopping times taking finitely many values in $\Q_{+}\cup\{\infty\}$. Indeed, let $\tau \in \cT_{+}$ be given; then 
 $$
 \tau_{n}:=\inf\big\{(k+1)2^{-n}:\, 0\leq k\leq n2^{n},\, \tau\leq k2^{-n}\big\}
 $$
 (where $\inf\emptyset=\infty$) is a sequence of such stopping times and $\tau_{n}\downarrow \tau$. Now  $A\cap \{\tau_{n}<\xi\}$ increases to $A\cap \{\tau<\xi\}$ for $A\in \cF_{\tau+}\subset \cF_{\tau_{n}}$; therefore, if~\eqref{eq: prior to zeta equiv} is valid for each $\tau_{n}$, then $ P(A\cap \{\tau<\xi\})=0$ implies  $ P(A\cap \{\tau_{n}<\xi\})=0$ which in turn implies  $ Q(A\cap \{\tau_{n}<\xi\})=0$ and thus $Q(A\cap \{\tau<\xi\})=0$ by monotone convergence.
 
Any $\F$-stopping time $\tau$ with finitely many values in $\Q_{+}\cup\{\infty\}$ is of the form $\tau=\sum_{i=1}^{n} t_{i}\1_{A_{i}}$, where $n\in \N$, $t_{i}\in\Q_{+}\cup\{\infty\}$ and $A_{i}\in \cF_{t_{i}}$ are disjoint. Hence,  
 $$
 R(A\cap \{\tau<\xi\})=\sum_{i=1}^{n} R\big(A\cap \{\tau\le t_{i}\}\cap A_{i}\cap  \{t_{i}<\xi\}\big),\quad R\in \{P,Q\}
 $$
 %NOTE: only the finite $t_{i}$ matter here
 and it follows that~\eqref{eq: prior to zeta equiv coutable} implies~\eqref{eq: prior to zeta equiv}.
\end{proof}

\begin{remark} \label{rem:localization_under_both}
Let $Q \simz P$. It is   a consequence of Lemma \ref{lem : PXa countable representation} that $Q$ and $P$ are equivalent on $\Fc_{\tau+} \cap \set{\tau < \zeta}$ for any $ {\tau} \in \Stop$. Suppose now that $(\tau_n)_{\nin}$ is a nondecreasing $\Stop$-valued sequence such that $\tau \dfn \limn \tau_n \geq \zeta$ holds in the $Q$-a.s. sense. Since $\set{\tau < \zeta} \in \Fc_{\tau+} \cap \set{\tau < \zeta}$ has zero $Q$-measure, we conclude that $P\{\tau < \zeta\} = 0$, i.e., that $\tau \geq \zeta$ also holds in the $P$-a.s.~sense. In particular, if $\zeta =  \infty$ $P$-a.s., it follows that $\tau =  \infty$ $P$-a.s.
\end{remark}

\begin{remark}\label{rem : density process constructed from Q}
Let $P$ and $Q$ be two probability measures  on $(\Omega, \Fc)$ with   $Q \domz P$ and $\zeta =  \infty$ $P$-a.s. By utilizing appropriate versions of the Radon--Nikodym theorem and a c\`adl\`ag modification procedure, one may establish the existence of a $P$-a.s.\   nonnegative  c\`adl\`ag adapted process $Y$ such that
\begin{equation} \label{eq: measure_change}
	Q (A_\tau \cap \set{\tau < \zeta}) = E^P \bra{Y_\tau \1_{A_\tau} \1_{\tau < \zeta}} \quad \mbox{for all}\quad  \tau \in \Stop_+  \text{ and } A_\tau \in \Fc_{\tau+}.
\end{equation}
The above process $Y$ will be called the \emph{prior-to-$\zeta$ density process of $Q$ with respect to $P$}.   It is strictly positive under $P$ when $Q\simz P$.   Note that \eqref{eq: measure_change} uniquely specifies $Q$, since  the class  of sets $ A_T \cap \{ T < \zeta  \} , T \in \Real_+ , A_T \in \Fc_T $    generates $\Fc_{\zeta-} = \Fc$ and   is also a $\pi$-system. Therefore, the specification  of the  prior-to-$\zeta$ density process of $Q$ with respect to $P$ is uniquely defined up to $P$-evanescent sets.

Suppose that $Q \simz P$ and $Y$ is the prior-to-$\zeta$ density process of $Q$ with respect to $P$. In particular, since $Q$ and $P$ are equivalent on $\Fc_{0+}$ and $\zeta > 0$, \eqref{eq: measure_change} gives $E^P[Y_0] = 1$. Furthermore, for $0 \leq s < t<\infty$ and $A_s \in \Fc_{s+}$, note that
\[
E^P [Y_t \1_{A_s}] = Q (A_s \cap \set{t < \zeta}) \leq Q (A_s \cap \set{s < \zeta}) = E^P [Y_s \1_{A_s}],
\]
which implies that $Y$ is an $(\bF_+, P)$-supermartingale. 
\end{remark}

Theorem~\ref{thm: foelmeasure_1} that follows, essentially due to F\"ollmer in \cite{Foellmer.72}, is a converse to the previous observation: starting with a probability $P$ and a candidate density process $Y$, a probability $Q$ is constructed that has $Y$ as a prior-to-$\zeta$ density with respect to $Q$. The statement requires the following notion.

\begin{definition} \label{defn: foretellable}
We say that $\zeta$ is \emph{foretellable} under a probability $Q$ if there exists a $\Stop_+$-valued sequence $(\tau_n)_{\nin}$ such that $Q\{\tau_n < \zeta\} = 1$ for all $n$ and $Q\{\limn \tau_n = \zeta\} = 1$.
\end{definition}

It is clear that the above sequence of stopping times can be chosen to be nondecreasing. Also, note that  foretellability of $\zeta$  does \emph{not} remain invariant under prior-to-$\zeta$ equivalent probability changes.

\begin{remark} \label{rem:foretell_and_pred}
By \cite[Theorem 4.16]{HeWangYan.92}, $\zeta$ is foretellable under $Q$ if and only if $\zeta$ is $Q$-a.s.~equal to a predictable stopping time on $(\Omega, \, \bF_+)$. 
\end{remark}

%[Note that in our case we do not have to add F_0$; since $\zeta > 0$ everywhere, $\F_0$ is already included in $\F_{\xi-}$ as defined previously.]

\begin{theorem} \label{thm: foelmeasure_1}
Let $Y$ be a strictly positive $(\bF_+, P)$-supermartingale with $E^P \bra{Y_0} = 1$. Then, there exists $Q \simz P$ such that $Y$ is the prior-to-$\zeta$ density process of $Q$ with respect to $P$. Furthermore, if $Y$ is actually an $(\bF_+, P)$-local martingale, $\zeta$ is foretellable under $Q$.
\end{theorem}

\begin{proof}
Recall that for $\xi \in \Stop_+$, the $\sigma$-field $\Fc_{\xi-}$ is generated by the collection $\set{A_s \cap \set{s < \xi}:\,  s \geq 0,\, A_s \in \Fc_s}$. With this definition in place, we observe that $\Fc = \Fc_{\zeta-}$, because $B_t$ is $\Fc_{\zeta-}$-measurable for all $t \geq 0$. Indeed, Borel subsets of $E \cup \{ \triangle \}$ are of the form $A$ or $A \cup \{ \triangle \}$, where $A \in \cB(E)$, and for any such $A$, we have $\{ B_t \in A \} = \{ B_t \in A \}  \cap \{ t < \zeta \} \in \cF_{\zeta -}$  and $\set{B_t \in A \cup \{ \triangle \}} = (\{ B_t \in A \}  \cap \{ t < \zeta \}) \cup \{ \zeta \leq t \} \in \Fc_{\zeta -}$. 

By \cite[Section 4.2]{PerkowskiRuf.13}, one can construct $\xi \in \Stop_+$ with $P\{\xi < \infty\} = 0$ and a probability $Q^0$ on $(\Omega, \Fc_{\xi-})$, such that 
$$
Q^0 (A_\tau \cap \set{\tau < \xi}) = E^P \bra{Y_\tau \1_{A_\tau} \1_{\tau < \xi}}
$$
 holds for all $\tau \in \Stop_+$ and $A_\tau \in \Fc_{\tau+}$. In particular, $Q^{0}\{\xi >0\}=E^{P}[Y_{0}]=1$. Since $A_\tau \cap \set{\tau < \xi\wedge\zeta} \in \Fc_{\tau+}$ for all $A_\tau \in \Fc_{\tau+}$, the above formula also holds for $\xi':=(\xi\wedge \zeta) \1_{\xi > 0} + \zeta \1_{\xi = 0}$. Thus, we may assume that $\xi \in \Stop_+$ satisfies $0 < \xi \leq \zeta$ and $P\{\xi = \zeta\} = 1$, and that $Q^0 (A_\tau \cap \set{\tau < \xi}) = E^P \bra{Y_\tau \1_{A_\tau} \1_{\tau < \xi}}$ holds for all $\tau \in \Stop_+$ and $A_\tau \in \Fc_{\tau+}$.
We shall extend $Q^0$ to a probability $Q$ on $\Fc = \Fc_{\zeta-}$ such that $Q\{\xi = \zeta\} = 1$ holds; this will immediately establish~\eqref{eq: measure_change}. Define a map $\psi: \Omega \rightarrow \Omega$ as follows: for $\omega \in \Omega$, set $\psi_{t}(\omega) = \omega_{t}$ when $t < \xi (\omega)$ and $\psi_{t}(\omega) =\triangle$ when $\xi (\omega) \leq t$. Since $\Fc$ is generated by the coordinate projections and 
\[
\{ \psi \in \Lambda \} = ( \{\omega :\, \omega_{t} \in \Lambda \cap E  \} \cap \{ t < \xi \} ) \cup \{t \geq \xi \}  \in \Fc_{\xi-}
\]
holds for all $\tir$ and Borel subsets $\Lambda$ of $\bar{E}=E\cup \{ \triangle\}$, it follows that $\psi$ is $(\Fc_{\xi-} / \Fc)$-measurable. By construction,  $\zeta \circ \psi = \xi$. We claim that $\xi \le \xi\circ \psi$ holds as well. Indeed, since $\xi\wedge t$ is $\cF_{t-}$-measurable for all $\tir$, \cite[Theorem~96, Chapter~IV]{DellacherieMeyer.78} implies that $\xi \wedge t = (\xi \wedge t)\circ k_{t}$, where $k$ is the killing operator defined via $k_{t}(\omega) = \omega\1_{[0,t)} + \triangle \1_{[t,\infty)}$ for $\omega \in \Omega$. Since $\xi(\omega)\wedge t= \xi\circ k_{t}(\omega) \wedge t$ holds for all $(\omega, t) \in \Omega \times \Real_+$, plugging in $t = \xi(\omega)$ gives
$$
\xi(\omega)= \xi\circ k_{\xi(\omega)}(\omega) \wedge \xi(\omega)=\xi\circ \psi (\omega) \wedge \xi(\omega), \quad   \omega \in \Omega,
$$
where we have used that $k_{\xi(\omega)}(\omega) = \psi(\omega)$ holds for all $\omega \in \Omega$. Therefore,  $\xi \le \xi\circ \psi$. The last inequality, combined with $\xi \leq \zeta$ and $\zeta \circ \psi = \xi$, gives $\zeta \circ \psi = \xi  \circ \psi$. Define $Q$ on $\Fc$ via $Q(A )= Q^0(\psi^{-1}(A))$  for all $A \in \Fc $. By construction, $Q$ is an extension of $Q^0$, and \eqref{eq: measure_change} follows since 
$$
Q\{\xi < \zeta\} = Q^0\{\xi(\psi) < \zeta(\psi)\} = Q^0(\emptyset) = 0.
$$

Finally, if $Y$ is an $(\bF_+, P)$-local martingale, let $(\tau_n)$ be a localizing sequence and call  $\tau \dfn \limn \tau_n$. Note that $\tau = \infty = \zeta$ holds in the $P$-a.s.\ sense. By Remark \ref{rem:localization_under_both}, $\tau\geq \zeta$ holds in the $Q$-a.s.\ sense. Furthermore, from~\eqref{eq: measure_change}, we obtain $Q\{\tau_n < \tau\} = E_P [Y_{\tau_n}] = 1$ for all $\nin$. Therefore, $\zeta$ is foretellable under $Q$.
\end{proof}

\subsection{On the Path Space $\Omega$}

The goal of this section is to show that $\Omega$ carries a natural Polish topology, which is required for the measurable selection arguments in Sections~\ref{sec: dyna prog and local mart} and~\ref{sec: superhedging duality}. To the best of our knowledge, this result is not contained in the previous literature---only the Lusin property is mentioned; see, e.g., \cite{Meyer.72}.

Let $\D=\D_{x_{*}}([0,\infty);E)$ be the usual Skorokhod space of $E$-valued c\`adl\`ag paths on $[0,\infty)$ starting at the point $x_{*}\in E$ and let $\delta_{\infty}$ be its usual metric, rendering $\D$ a Polish space. We may think of a path $\omega\in\Omega$ as consisting of a path $\tilde\omega\in\D$ and a lifetime $z\in(0,\infty]$; in this context, it is useful to equip $(0,\infty]$ with the complete metric $d_{(0,\infty]}(z,z'):=|z^{-1}-z'^{-1}|$, where $\infty^{-1}:=0$.
More precisely, given $z\in(0,\infty]$, let
$$
e_{z}(t):=\begin{cases}
 t&\mbox{if }z=\infty,\\
 z(1-e^{-t})&\mbox{if }z<\infty.
\end{cases} 
$$
We note that $e_{z}: [0,\infty)\to[0,z)$ is a monotone bijection; thus, precomposition with $e_{z}$ turns a path $\omega\in\Omega$ with lifetime $z=\zeta(\omega)$ into an element of~$\D$.
As a result, we can define
$$
\delta_{\Omega}(\omega,\omega'):= d_{(0,\infty]}\big(\zeta(\omega),\zeta(\omega')\big) + \delta_{\infty}\big(\omega\circ e_{\zeta(\omega)}, \omega'\circ e_{\zeta(\omega')}\big), \quad \omega,\omega'\in\Omega.
$$

\begin{lemma}\label{le:OmegaPolish} 
  The space $(\Omega,\delta_{\Omega})$ is Polish and its Borel $\sigma$-field coincides with~$\cF$. Moreover, $\cF_{\tau}=\sigma(B_{t\wedge\tau},\,t\in\R_{+})$ for any $\F$-stopping time $\tau$; in particular, $\cF_{\tau}$ is countably  generated.
\end{lemma}

\begin{proof}
  It is clear that $\delta_{\Omega}$ defines a metric on $\Omega$. Moreover, the mapping
  $$
    \Omega \to \D\times (0,\infty],\quad \omega \mapsto \big(\omega \circ e_{\zeta(\omega)}, \zeta(\omega)\big)
  $$
  admits the inverse 
  $$
     \D\times (0,\infty] \to \Omega,\quad (\tilde\omega,z) \mapsto (\tilde \omega \circ e_{z}^{-1}) \,\1_{[0,z)} + \triangle \,\1_{[z,\infty)}.
  $$
  By the definition of $\delta_{\Omega}$, these mappings constitute an isometric isomorphism between $\Omega$ and $\D\times (0,\infty]$; in particular, $\Omega$ is Polish like $\D\times (0,\infty]$.
%  To see that $\Omega$ is separable, let $(\tilde\omega_{n})_{n\geq1}$ be a dense sequence in $\D$ and let $(z_{m})_{m\geq1}$ be a dense sequence in $(0,\infty]$. Then
%  $$
%  \omega^{m}_{n}:= (\tilde\omega_{n}\circ e_{z_{m}}^{-1}) \,\1_{[0,z_{m})} + \triangle \,\1_{[z_{m},\infty)},\quad m,n\geq1
%  $$
%  defines a countable dense set in $\Omega$. To show that $\Omega$ is complete, let $(\omega_{n})_{n\geq1}$ be a Cauchy sequence in $\Omega$ and set $z_{n}:=\zeta(\omega_{n})$. Then $z_{n}$ is Cauchy in $(0,\infty]$ and thus converges to some $z\in (0,\infty]$. Moreover, $\omega_{n}\circ e_{z_{n}}$ is Cauchy in $\D$ and thus has a limit $\tilde \omega\in\D$. We then see that
%  $$
%    \omega:=(\tilde \omega \circ e_{z}^{-1}) \,\1_{[0,z)} + \triangle \,\1_{[z,\infty)}
%  $$
%  is an element of $\Omega$ such that $\delta_{\Omega}(\omega_{n},\omega)\to 0$. This completes the proof that $(\Omega,\delta_{\Omega})$ is Polish. 
  
  Let $\cB(\Omega)$ be the Borel $\sigma$-field on $\Omega$. To prove that $\cF\subset\cB(\Omega)$, it suffices to show that the evaluation 
  $B_{t}: \omega\mapsto\omega_{t}$ is Borel-measurable for any fixed $t\geq0$. To this end, note that the functions  
  $$
  \omega\mapsto \zeta(\omega)\in (0,\infty], \quad \omega\mapsto \omega\circ e_{\zeta(\omega)}\in \D,\quad \omega\mapsto e^{-1}_{\zeta(\omega)}(t)\in[0,\infty)
  $$
  are continuous on $\Omega$. Let $\tilde B$ be the canonical process on $\D$ and recall that $(t,\tilde\omega)\mapsto \tilde B_{t}(\tilde\omega)$ is jointly Borel-measurable. It then follows that
  $$
    \omega\mapsto B_{t}(\omega)= \tilde{B}_{e^{-1}_{\zeta(\omega)}(t)}\big(\omega\circ e_{\zeta(\omega)}\big) \,\1_{[0,\zeta(\omega))}(t) + \triangle \,\1_{[\zeta(\omega),\infty)}(t)
  $$
  is Borel-measurable as well.
  
  To prove the reverse inclusion $\cB(\Omega)\subset\cF$, it suffices to show that any continuous function $f:\Omega\to\R$ is $\cF$-measurable. 
  Indeed, the maps
  $$
  \omega\mapsto \zeta(\omega)\in (0,\infty], \quad \omega\mapsto \omega\circ e_{\zeta(\omega)}\in \D
  $$
  are clearly $\cF$-measurable. Moreover, any function $f$ on $\Omega$ induces a unique function $\tilde f$ on $\D\times (0,\infty]$ satisfying
  $$
    f(\omega) = \tilde f\big(\omega\circ e_{\zeta(\omega)}, \zeta(\omega)\big),\quad \omega\in\Omega.
  $$
  If $f$ is continuous, it follows that $\tilde f$ is  continuous and hence the composition 
  $\omega\mapsto f(\omega)=\tilde f(\omega\circ e_{\zeta(\omega)}, \zeta(\omega))$ is $\cF$-measurable. This completes the proof that $\cF=\cB(\Omega)$.
  
  The last claim follows from the fact that $\bar E$ is Polish and standard arguments; see \cite[Lemma~1.3.3 and Exercise~1.5.6]{StroockVaradhan.79}.
\end{proof}

%%%%%%%%%%%%%%%%%%%%%%%%%%%%%%%%%%%%%%%%%%%%%%%%%%%%%%%%%%
%\bibliography{stochfin}
%\bibliographystyle{plain}
%%%%%%%%%%%%%%%%%%%%%%%%%%%%%%%%%%%%%%%%%%%%%%%%%%%%%%%%%

\newcommand{\dummy}[1]{}

\end{document}